\documentclass[11pt]{amsart}
\usepackage{fullpage}
\usepackage[foot]{amsaddr}
\usepackage{microtype}
\usepackage[OT1]{fontenc}
\usepackage{amsmath}
\usepackage{amssymb}
\usepackage{amsthm}
\usepackage[linesnumbered,boxed,ruled,vlined]{algorithm2e}
\usepackage{algpseudocode}
\usepackage{enumitem}
\usepackage{xifthen}
\usepackage{xspace}

\usepackage[margin=1cm]{caption} % Adds an addition margin on either side of figure captions. This must come before "\usepackage{subfig}" (otherwise there are errors).
\usepackage{subfig}

\usepackage[thinlines]{easytable}

\usepackage[bookmarks=true,hypertexnames=false,pagebackref]{hyperref}
\hypersetup{colorlinks=true, citecolor=blue, linkcolor=red, urlcolor=blue}

\usepackage{pgfplots}
\pgfplotsset{compat=1.16}
\usepackage{tikz}
\usetikzlibrary{arrows,arrows.meta,backgrounds,calc,fit,decorations.pathreplacing,decorations.markings,shapes.geometric}

\tikzstyle{internal} = [draw, fill, shape=circle]
\tikzstyle{external} = [shape=circle]
\tikzstyle{square}   = [draw, fill, rectangle]
\tikzstyle{triangle} = [draw, fill, regular polygon, regular polygon sides=3, inner sep=3pt]
\tikzstyle{pentagon} = [draw, fill, regular polygon, regular polygon sides=5, inner sep=2pt, minimum size=14pt]
\tikzset{every fit/.append style=text badly centered}

\usetikzlibrary{positioning,chains,fit,shapes,calc}
\usetikzlibrary{trees}
\usetikzlibrary{decorations.pathreplacing}
\usetikzlibrary{decorations.pathmorphing}
\usetikzlibrary{decorations.markings}
\tikzset{>=latex} % arrow tips
\usepackage{ifthen}

\usepackage{cleveref}

\usepackage[textsize=tiny]{todonotes}

\usepackage[normalem]{ulem}

\usepackage{mleftright}

\usepackage{cool}
\Style{DSymb={\mathrm d},DShorten=true,IntegrateDifferentialDSymb=\mathrm{d}}

\newcommand{\tp}[1]{{\left( #1 \right)}}

\newcommand{\Ex}{\mathop{\mathbb{{}E}}\nolimits}
\renewcommand{\Pr}{\mathop{\mathrm{Pr}}\nolimits}

\def\*#1{\mathbf{#1}}
\def\+#1{\mathcal{#1}}
\def\-#1{\mathrm{#1}}
\def\=#1{\mathbb{#1}}
\def\^#1{\mathbb{#1}}

\newcommand{\abs}[1]{\ensuremath{\left\vert#1\right\vert}}

\newcommand{\ceil}[1]{\lceil#1\rceil}

\newcommand{\set}[1]{\left\{#1\right\}}
\newcommand{\eps}{\varepsilon}

\newcommand{\dist}{\operatorname{dist}}

\newcommand{\defeq}{:=}

\newtheorem{theorem}{Theorem}

\newtheorem{lemma}[theorem]{Lemma}

\newtheorem{observation}[theorem]{Observation}
\newtheorem{proposition}[theorem]{Proposition}
\newtheorem{corollary}[theorem]{Corollary}
\theoremstyle{definition}
\newtheorem{condition}[theorem]{Condition}
\newtheorem{definition}[theorem]{Definition}

\theoremstyle{remark}
\newtheorem*{remark}{Remark}

\crefname{theorem}{Theorem}{Theorems}
\crefname{observation}{Observation}{Observations}
\crefname{claim}{Claim}{Claims}
\crefname{condition}{Condition}{Conditions}
\crefname{algorithm}{Algorithm}{Algorithms}
\crefname{property}{Property}{Properties}
\crefname{example}{Example}{Examples}
\crefname{fact}{Fact}{Facts}
\crefname{lemma}{Lemma}{Lemmas}
\crefname{corollary}{Corollary}{Corollaries}
\crefname{definition}{Definition}{Definitions}
\crefname{remark}{Remark}{Remarks}
\crefname{proposition}{Proposition}{Propositions}
\crefname{equation}{equation}{equations}
\crefname{enumi}{Case}{Case}
\creflabelformat{enumi}{(#2#1#3)}

\makeatletter
\def\prob#1#2#3{\goodbreak\begin{list}{}{\labelwidth\z@ \itemindent-\leftmargin
      \itemsep\z@  \topsep6\p@\@plus6\p@
      \let\makelabel\descriptionlabel}
  \item[\textbf{Name}]#1
  \item[\textbf{Instance}]#2
  \item[\textbf{Output}]#3
  \end{list}}
\makeatother

\makeatletter
\providecommand\@dotsep{5}
\def\listtodoname{Todo list}
\def\listoftodos{\@starttoc{tdo}\listtodoname}
\makeatother

\newcommand{\dTV}{d_{\mathrm{TV}}}

\usepackage{nicefrac,comment}
\newboolean{doubleblind}
\setboolean{doubleblind}{false}

\title{Deterministic counting from coupling independence}

\ifdoubleblind
\author{Author(s)}
\else

\author{Xiaoyu Chen, Weiming Feng, Heng Guo, Xinyuan Zhang, Zongrui Zou}

\address[Xiaoyu Chen, Xinyuan Zhang, Zongrui Zou]{State Key Laboratory for Novel Software Technology, New Cornerstone Science Laboratory, Nanjing University, 163 Xianlin Avenue, Nanjing, Jiangsu Province, China.}
\address[Weiming Feng]{School of Computing and Data Science, The University of Hong Kong, Pokfulam Road, Hong Kong, China.}
\address[Heng Guo]{School of Informatics, University of Edinburgh, Informatics Forum, Edinburgh, EH8 9AB, United Kingdom.}
\thanks{This project has received funding from the European Research Council (ERC) under the European Union's Horizon 2020 research and innovation programme (grant agreement No.~947778). Weiming Feng acknowledges the support from Dr.~Max R\"ossler, the Walter Haefner Foundation and the ETH Z\"urich Foundation during his affiliation with ETH Z\"urich.}

\fi

\begin{document}

\begin{abstract}
  We show that spin systems with bounded degrees and coupling independence admit fully polynomial time approximation schemes (FPTAS).
  We design a new recursive deterministic counting algorithm to achieve this. 
  As applications, we give the first FPTASes for $q$-colourings on graphs of bounded maximum degree $\Delta\ge 3$, when $q\ge (11/6-\eps_0)\Delta$ for some small $\eps_0\approx 10^{-5}$,
  or when $\Delta\ge 125$ and $q\ge 1.809\Delta$,
  and on graphs with sufficiently large (but constant) girth, when $q\geq\Delta+3$.
  These bounds match the current best randomised approximate counting algorithms
  by Chen, Delcourt, Moitra, Perarnau, and Postle (2019), Carlson and Vigoda (2024), and Chen, Liu, Mani, and Moitra (2023), respectively.
\end{abstract}
\maketitle

\section{Introduction}

The power of randomness is a classical topic in the theory of computing.
Randomised algorithms have found many early successes in the field of approximate counting.
A striking example is the polynomial-time volume estimation algorithm for convex bodies by Dyer, Frieze, and Kannan \cite{DFK91},
whereas deterministic approximation algorithms requires at least exponential membership queries \cite{Ele86,BF87}.
However, this lower bound is valid only for membership query models, and does not rule out efficient deterministic approximation algorithms in general.
While volume estimation remains difficult for efficient deterministic approximation,
deterministic approximate counting algorithms have been quickly catching up with their randomised counterparts for many other problems,
since the introduction of the correlation decay technique \cite{weitz2006counting,BG06}.

By now, a number of different deterministic approximate counting techniques have been developed.
In addition to the correlation decay method, 
one may utilise zero-freeness of polynomials \cite{Bar16,PR17}, linear programming based methods \cite{Moi19,GLLZ19,JPV21}, statistical physics related techniques \cite{HPR20,JPP23,JPSS22}, 
or even direct derandomisation of Markov chains \cite{FGWWY23}.
In many occasions, these methods have achieved optimal results, or at least match or even outperform the best randomised algorithms,
such as for the hardcore gas model \cite{weitz2006counting,PR17}, for Holant problems with log-concave signatures \cite{HLQZ25}, or in the local lemma settings \cite{WY24}. 

Despite all these successes, there is one problem where deterministic algorithms are still lagging behind,
namely counting the number of proper colourings.
The study of this problem was initiated by Jerrum \cite{Jer95},
who showed a rapid mixing bound for Glauber dynamics for $q$-colourings on graphs of maximum degree $\Delta$, when $q>2\Delta$.
This was subsequently improved by Vigoda \cite{Vig00} to $q>11/6\Delta$ by considering the flip dynamics.
Via more careful analysis, the constant was then improved to $(11/6-\eps_0)$ for some $\eps_0\approx10^{-5}$ by Chen, Delcourt, Moitra, Perarnau, and Postle \cite{chen2019improved},
and to $1.809$ by Carlson and Vigoda \cite{CV24} for $\Delta\ge 125$.
In contrast, on the deterministic side, the first efficient algorithm by Gamarnik and Katz \cite{GK07} requires $q>2.844\Delta$ via the correlation decay method.
This was later improved by Lu and Yin \cite{LY13} to $q>2.581\Delta$.
With the more recent technique using zero-freeness of polynomials, Liu, Srivastava, and Sinclair \cite{LSS19} gave a fully polynomial time approximation schemes (FPTAS) when $q\ge 2\Delta$,
and this bound was very recently improved to $q\ge (2-\eps_1)\Delta$ for some $\eps_1\approx 0.002$ by Bencs, Berrekkal, and Regts \cite{BBR24}.

In this paper, we close this gap between deterministic and randomised algorithms for approximate counting colourings.
We introduce a new algorithm that takes inspiration from both the linear programming method and the correlation decay method.
We show that this algorithm is efficient as long as coupling independence holds and there is a marginal lower bound.
Here coupling independence is a method to establish the so-called spectral independence \cite{ALO20,AL20}, a relatively new tool to analyse mixing times of Markov chains.
The notion of coupling independence is formally introduced in \cite{CZ23}, although it has been implicitly established before that, such as in \cite{FGYZ22,blanca2022mixing,liu2021coupling}.
Previously, coupling independence is mainly used to analyse Markov chains,
and here we show that it also implies deterministic approximate counting algorithms.
We also show that contractive coupling for Markov chains can be used to establish coupling independence.
Thus, with our technique, the contractive couplings from \cite{chen2019improved,CV24} imply FPTASes with matching bounds.
We describe our main results in more detail in \Cref{sec:results},
and give a high-level technical overview in \Cref{sec:techniques}.

%Starting from the rapid mixing of bases-exchange walks \cite{ALOV24,CGM21,ALOVV21},
%there has been great progress on analysing mixing times of Markov chains recently.
%One of the main technique behind this wave of results is spectral independence \cite{ALO20,AL20}, 
%which has found numerous applications such as \cite{CLV21,CLV23,CLV24,FGYZ22,CGSV21,CFYZ21,blanca2022mixing,AASV21,AJKPV22,ALV22,CFYZ22,CZ23,EF23,CLMM23,Chen24,AJKPV24,CG24,WZZ24,CF24}.
%The main ways to establish spectral independence include: via trickle-down type theorems \cite{Opp18,ALO21,WZZ24}, geometry of polynomials \cite{AASV21,CLV24}, and coupling independence \cite{CZ23,CG24,CF24}.

\subsection{Main results}\label{sec:results}

We state our main results in the general context of spin systems.
A spin system is specified by the tuple $\mathcal{S}=(G, q, A_E, A_V)$.
Given a graph $G=(V,E)$ and an integer $q>0$,
a state of the system is a configuration $\sigma: V\rightarrow [q]$.
Namely, the state space is $[q]^V$.
The weight of a configuration are characterised by the matrix $A_E\in\mathbb{R}^{q\times q}_{\ge 0}$ and vector $A_V\in\mathbb{R}^{q}_{\ge 0}$.
The Gibbs distribution $\mu$ of $\mathcal{S}$ is defined by
\begin{align*}
  \mu(\sigma) \propto w(\sigma) \defeq \prod_{\{u,v\}\in E}A_E(\sigma(u),\sigma(v))\prod_{v\in V}A_V(\sigma(v)).
\end{align*}
The normalising factor of $\mu$, namely the so-called partition function of $\mathcal{S}$, is defined by
\begin{align*}
  Z \defeq \sum_{\sigma \in [q]^V}w(\sigma).
\end{align*}
When $A_E= \begin{pmatrix}
  1 & 1 \\
  1 & 0
\end{pmatrix}$ and $A_v=  \begin{pmatrix}
  1  \\
  \lambda
\end{pmatrix}$, 
this encodes the hardcore gas model.
When $A_E=J-I$, where $J$ is the all-$1$ matrix and $I$ is the identity matrix, and $A_v=\*1$ is the all-$1$ vector,
$\mu$ is uniform over proper $q$-colourings,
and $Z$ is the number of them.
%Also, edge-valued problems, such as matchings or Holant problems, are spin systems on their line graph.

To introduce coupling independence, we need to define Hamming and Wasserstein distances.
For two configurations $\sigma$ and $\tau$, let their Hamming distance be
\begin{align*}
  \dist\{\sigma,\tau\}\defeq \abs{\{v\mid v\in V,~\sigma(v)\neq\tau(v)\}}. %\sum_{v\in V} \mathbb{1}_{\sigma(v)\neq\tau(v)}.
\end{align*}
The Wasserstein distance is defined next.
\begin{definition}  \label{def:Wasserstein}
  Let $(\Omega, \-d)$ be a finite metric space.
  For any two distributions $\mu$ and $\nu$ on $\Omega$, the \emph{$1$-Wasserstein distance} (\emph{W1-distance}) with respect to the metric $\-d$ between $\mu$ and $\nu$ is defined as
  \begin{align} \label{eq:def-W1-dis}
    \+W_{\-d}(\mu, \nu) := \inf_{\+C} \Ex_{(X,Y)\sim\+C}[\-d\tp{X,Y}],
  \end{align}
  where the infimum is taken over all the possible couplings $\+C$ between $\mu$ and $\nu$.
\end{definition}
For two $\Omega$-valued random variables $X, Y$ with distribution $\mu,\nu$, we may also use $\+W_{\-d}(X,Y)$ to denote $\+W_{\-d}(\mu,\nu)$.
When the distance $\-d$ is the Hamming distance, we also omit the subscript and write $\+W(\mu,\nu)$.

\begin{definition}[Coupling independence] \label{def:CI}
  We say a Gibbs distribution $\mu$ satisfies \emph{$C$-coupling independence} if, for any two partial configurations $\sigma$ and $\tau$ on $\Lambda\subseteq V$ such that $\dist\{\sigma,\tau\}=1$,
%  there is a coupling $\+C$ between $\mu^{\sigma}$ and $\mu^{\tau}$ such that
%  \begin{align*}
%    \Ex_{(\sigma',\tau')\sim\+C}\left[ \dist\{\sigma',\tau' \}\right] \le C.
%  \end{align*}
  \begin{align*}
    \+W(\mu^{\sigma},\mu^{\tau})\le C,
  \end{align*}
  where $\mu^{\sigma}$ and $\mu^{\tau}$ denote the Gibbs distribution conditional on $\sigma$ and $\tau$, respectively.
\end{definition}

In addition to coupling independence (CI), our main theorem also requires marginal lower bound.
Let $\mu^\sigma_v$ be the marginal distribution at $v$ conditional on $\sigma$.

\begin{definition}[Marginal lower bound]\label{def:marginal-bound}
  We say a Gibbs distribution $\mu$ over $[q]^V$ is $b$-marginally bounded if for any partial configuration $\sigma \in [q]^\Lambda$ on $\Lambda \subseteq V$, any vertex $v \notin \Lambda$, any spin $c \in [q]$ with $\mu^\sigma_v(c) >0$,
  \begin{align}
    \mu^\sigma_v(c) \geq b.
  \end{align} 
\end{definition}

Now we are ready to state our main theorem.

\begin{theorem}\label{thm:main}
  Let $q \geq 2,b > 0,C > 0, \Delta \geq 3$ be constants. There exists a deterministic algorithm such that given a permissive spin system $\mathcal{S}=(G,q,A_E,A_V)$ and error bound $0 < \eps < 1$, if the Gibbs distribution of $\mathcal{S}$ is $b$-marginally bounded and satisfies $C$-coupling independence, and the maximum degree of $G$ is at most $\Delta$, then it returns $\hat{Z}$ satisfying $(1-\eps)Z \leq \hat{Z} \leq (1+\eps)Z$ in time $(\frac{n}{\eps})^{f(q,b,C,\Delta)}$,
  where $f(q,b,C,\Delta)=\Delta^{O(C(\log b^{-1} + \log C +\log\log\Delta))} \log q$ is a constant.
\end{theorem}

Being permissive is a mild technical condition we need (see~\Cref{def:permissive}).
It roughly requires that all the conditional distributions are well-defined.
All applications considered in this paper satisfy it.
The marginal lower bound is also a mild requirement, since if we treat $q$, $\Delta$, $A_E$, and $A_V$ all as constants,
then there is a constant $b$ such that any permissive system is $b$-marginally bounded (see \Cref{obs:permissive-lb}).
However, for concrete systems, there are usually better lower bound than the generic one in \Cref{obs:permissive-lb}.
Thus, we choose to make it explicit in the statement of \Cref{thm:main}.
We also note that the exponent $f(q,b,C,\Delta)$ is not optimised -- our goal is to present the new algorithm as simple and as clearly as possible.

\Cref{thm:main} together with coupling independence from \cite{CLMM23} and the marginal lower bound \cite[Lemma 3]{GKM15}\footnote{The bound yields $b\ge q^{-1}\left( \frac{2}{3} \right)^{\Delta}$ in the setting of \Cref{cor:main-colouring-high}.} directly implies the following result.

\begin{corollary} [Colouring: high-girth graphs]  \label{cor:main-colouring-high}
  Let $q$ and $\Delta$ be two integers satisfying $\Delta \geq 3$ and $q \geq \Delta + 3$. There eixsts a constant $g_0 > 0$ depending only on $\Delta$ such that the following holds.   There exists an FPTAS for the number of proper $q$-colourings on graphs $G$ of maximum degree $\Delta$ and girth at least $g_0$.
\end{corollary}

The bound in \Cref{cor:main-colouring-high} matches the rapid mixing result by Chen, Liu, Mani and Moitra~\cite{CLMM23}.
Prior to our work, no FPTAS is known in this setting.

To establish coupling independence for colourings in general bounded degree graphs, we make use of contractive coupling, a tool typically used to bound the mixing time of Markov chains.
Given two copies of a Markov chain, a contractive coupling ensures that after a step, the expected distance decreases multiplicatively.
See \eqref{eq:contract-1} and \eqref{eq:contract-2} for some examples.
Previously, contractive couplings have been used to establish spectral independence \cite{blanca2022mixing,liu2021coupling}.
In fact, their proof implicitly established the stronger result of coupling independence.
In \Cref{sec:contractive-coupling}, we give a simpler and more direct argument on how to establish coupling independence from contractive coupling.
%In \Cref{sec:contractive-coupling}, we show that they can also be used to establish the stronger property of coupling independence.
Using the contractive couplings from \cite{CV24} and \cite{chen2019improved}, as well as the marginal lower bound \cite[Lemma 3]{LY13},\footnote{The bound yields $b\ge q^{-1}e^{-\frac{1}{\alpha-1}}$ in the setting of \Cref{cor:main-colouring}, where $\alpha$ is the constant in the assumption $q\ge \alpha\Delta$.} we have the following result.

\begin{corollary} [Colouring: general graphs]  \label{cor:main-colouring}
  Let $q$ and $\Delta$ be two integers satisfying either
  \begin{itemize}
    \item $\Delta \geq 125, q \geq 1.809\Delta$;
    \item or $\Delta \geq 3, q \ge (11/6-\eps_0)\Delta$ for some fixed parameter $\eps_0 \approx 10^{-5}$.
  \end{itemize}
  There exists an FPTAS for the number of proper $q$-colourings on graphs of maximum degree $\Delta$.
\end{corollary}

The bounds in \Cref{cor:main-colouring} are the same ones as the rapid mixing of Markov chains results by Carlson and Vigoda \cite{CV24} or by Chen, Delcourt, Moitra, Perarnau and Postle~\cite{chen2019improved}. 
As mentioned before, prior to our work, the best FPTAS \cite{BBR24} requires $q\ge (2-\eps_1)\Delta$ for some $\eps_1\approx 0.002$.

\Cref{thm:main} also implies FPTASes for spin systems satisfying the \emph{Dobrushin-Shlosman} condition.

\begin{definition}[Dobrushin-Shlosman condition]\label{def:DS}
 Let $\mu$ be a Gibbs distribution on $[q]^V$.
 The \emph{Dobrushin influence matrix} $\rho \in \={R}^{V\times V}_{\geq 0}$ is defined by
 \begin{align}\label{eqn:defDob}
   \forall u,v\in V, \rho(u,v) \defeq \max_{\substack{\sigma,\tau \in [q]^{V-v} \\ \sigma \oplus \tau \subseteq \set{u}}} \dTV(\mu^\sigma_v,\mu^\tau_v),
 \end{align}
 where we use $\dTV(\cdot,\cdot)$ to denote the TV-distance.
 The Gibbs distribution $\mu$ is said to satisfy the \emph{Dobrushin-Shlosman condition} with gap $\delta \in (0,1)$ if 
\begin{align*}
\Vert \rho \Vert_1 = \max_{u \in V} \sum_{v \in V} \rho(u,v) \leq 1 - \delta.
\end{align*}
%In particular, when $\mu$ is a Gibbs distribution for some graph $G = (V, E)$, we will have $R_{vu} = 0$ for non-adjacent $u, v$, by conditional independence.
\end{definition}

\begin{corollary} \label{cor:main-DS}
  Let $q \geq 2$, $\Delta \geq 3$, $A_E \in \mathbb{R}_{\geq 0}^{q \times q}$, $A_V \in \mathbb{R}_{\geq 0}^q$, and $\delta \in (0,1)$ be constant parameters. 
  There exists an FPTAS for the partition function of permissive spin systems $\mathcal{S}=(G, q, A_E, A_V)$ 
  if $\mathcal{S}$ satisfies the Dobrushin-Shlosman condition with gap $\delta$ and the maximum degree of $G$ is at most~$\Delta$.
\end{corollary}

The Dobrushin-Shlosman condition~\cite{dobrushin1970prescribing,dobrushin1985constructive} is a sufficient condition for the uniqueness of the Gibbs measure in infinite graphs. 
It is well-known that the Dobrushin-Shlosman condition implies rapid mixing of Glauber dynamics~\cite{BubleyD97,hayes2006simple}.
However, before our result, its implication on deterministic counting algorithms is not well-understood, especially for general multi-spin ($q > 2 $) systems. 

In addition, \Cref{thm:main} provides a unified framework to derive FPTASes for a few problems for which FPTASes are known before via different methods.
Similar to our method, all these FPTASes require a bounded maximum degree $\Delta$ for the input graph.
Examples in this category include $q$-colourings for triangle-free graphs if $q>1.764\Delta+C$ for some constant $C>0$ \cite{LSS19} (where CI is established in \cite{FGYZ22,CGSV21,CF24}),
antiferromagnetic two-state spin systems in the uniqueness regime \cite{SST14} (where CI is established in \cite{CF24,CLV23}),
ferromagnetic Ising models with non-zero external fields \cite{LSS19a} (where CI is established in \cite{CZ23}),
and Holant problems with log-concave signatures \cite{HLQZ25} (where CI is established in \cite{CG24}).\footnote{Technically Holant problems are not spin systems, but our technique generalises to that setting easily.}

After this paper appeared on the arXiv, Chen, Wang, Zhang, and Zhang \cite{CWZZ25} used \Cref{thm:main} to give an FPTAS for edge colourings when $q\ge 3\Delta$.

\subsection{Our technique} \label{sec:techniques}

We first show that coupling independence implies decay of total influences.
Here, the influence from $u$ to $v$ is defined as $\dTV({\mu^{\sigma}_v},{\mu^{\tau}_v})$ such that $\sigma$ and $\tau$ are \emph{partial configuations} that differ only at $u$. 
Note that this is different from Doburshin's influence defined in~\eqref{eqn:defDob}. See \Cref{def:TI-decay} for the formal definition of influences.
%The following theorem shows that the coupling independence implies total influence decay.
Suppose $C$-coupling independence holds. 
By a simple averaging argument, there is $1\le \ell\le 2C$ such that the total influence of a vertex $v$ to vertices of distance $\ell$ from $v$ is at most $1/2$.
We repeatedly use this fact to show that, roughly, for every $2C$ distance, the total influence decays multiplicatively by a factor $1/2$,
which leads to an exponential decay in the distance from $v$.
This is in \Cref{sec:influence-decay}.
Given the decay of total influence, we may choose a sufficiently large but constant $R$ such that the total influence at distance $R$ is less than a quantity that is linear in $R$. 
(To be precise, the requirement is given in \eqref{eqn:select-R}.)
This choice is a key guarantee in our error analysis.

Our main algorithm estimates marginal probabilities of a vertex $v$ under arbitrary conditioning.
The basic building block of our algorithm does this task but requires a constant number of other conditional marginal probabilities.
These marginal probabilities are either of the same vertex $v$ but with different boundary conditions at distance $R$,
or of some other vertices at distance $R$.
We use the aforementioned bounds for $R$ to show that the output relative error is roughly half of the relative error of the input marginals.
Thus, we achieved a recursive step with constant error decay.
To achieve an $\eps$ relative error, we just need to run this up to $\log\eps^{-1}$ depth,
resulting in polynomial total running time.

On a high level, the structure of our algorithm is very similar to the correlation decay algorithms such as \cite{weitz2006counting,BG06}.
Namely, in each step, we use the marginal probabilities in smaller instances to compute the desired one,
and the computation is truncated at logarithmic depth.
However, a key difference is that in previous algorithms, the recursive step is exact and often via a closed form formula, whereas ours is algorithmic and approximate.
This is achieved via a linear programming (LP) algorithm inspired by \cite{Moi19}.
We use the LP to simulate a coupling similar to the one in \cite{CLMM23},
which is to couple vertices randomly chosen at distance $R$ one at a time.
However, the coupling in \cite{CLMM23} is recursive, and it appears difficult for the LP to handle recursion.
We instead only use the LP to simulate a partial coupling up to the recursive point.
Because this is a partial coupling, certain quantities in the LP cannot be computed efficiently.
We rewrite these quantities in terms of marginal probabilities of smaller instances,
and use recursion to solve this issue.
%We use recursive calls for these quantities, or equivalently, these quantities are the input of our recursive step.
To control the evolution of relative errors, we need to find some new linear constraints that can be computed efficiently and characterise the failure probability of the coupling process.
These constraints, together with our choice of the radius $R$, ensure that the relative error decay by a constant factor each time.
The description and analysis of our algorithm are given in \Cref{sec:alg}.
For readers not familiar with Moitra's approach,
we also provide some heuristics and intuition behind it in \Cref{sec:heuristics}.

Comparing with all previous instantiation of the LP-based approximate counting algorithm \cite{Moi19,GLLZ19,JPV21,WY24,HLQZ25},
we do not write a polynomial-sized LP to solve.
Instead, our LP is only of constant size, but we recursively construct polynomially many of them.
Technically, because each of our sub-instances has only constant size, we are not obliged to use LP. 
(For example, we could write a quadratic system to solve.)
We choose LP just for technical convenience.
Moreover, because the (partial) coupling we use the LP to simulate has not been considered in this context before,
we need to write a new set of constraints to certify this coupling.

The coupling in \cite{CLMM23} establishes coupling independence for high-girth graphs when $q\ge \Delta+3$.
Thus our algorithm works in this setting, resulting in \Cref{cor:main-colouring-high}.
To apply our algorithm on general bounded degree graphs, we still need to establish coupling independence.
To this end, we show that contractive coupling for Markov chains (which is the main technique behind the rapid mixing results of \cite{chen2019improved} and \cite{CV24}) 
can be used to establish coupling independence.
This argument is given in \Cref{sec:contractive-coupling}.

\section{Preliminaries}\label{sec:prelim}
Let $\mu$ over $[q]^V$ be a Gibbs distribution on $G=(V,E)$ with matrix $A_E$ and vector $A_V$. 
Recall that $\mu$ is defined by for any configuration $\sigma \in [q]^V$,
\begin{align*}
  \mu(\sigma) = \frac{w(\sigma)}{Z}, \text{ where } w(\sigma) = \prod_{\{u,v\}\in E}A_E(\sigma(u),\sigma(v))\prod_{v\in V}A_V(\sigma(v)) \text{ and } Z = \sum_{\tau \in [q]^V}w(\tau).
\end{align*}
We call $w(\sigma)$ the weight of $\sigma$ and $Z$ the partition function.

For any subset $S\subseteq V$, we use subscript $S$ to denote the marginal distribution,
namely, for any $\tau\in[q]^S$, $\mu_S(\tau)\defeq \sum_{\tau'\in[q]^{V\setminus S}}\mu(\tau\cup\tau')$.
When $S=\{v\}$ is a singleton set, we may also write $\mu_v$ instead of $\mu_{\{v\}}$.

Let $\sigma \in [q]^{V \setminus \Lambda}$ be a partial configuration on $V \setminus \Lambda$, where $\sigma$ is allowed to be infeasible, namely we allow $\mu_{V \setminus \Lambda}(\sigma) = 0$. 
For any $\tau \in [q]^V$, define the conditional weight
\begin{align*}
    w^\sigma(\tau) =\*{1}_{\tau(V \setminus \Lambda)=\sigma} \cdot \prod_{v \in \Lambda}A_V(\tau(v)) \prod_{\{u,v\}\in E: u \in \Lambda \land v \not\in \Lambda} A_E(\tau(u),\sigma(v)) \prod_{\{u,v\}\in E:u \in \Lambda \land v \in \Lambda} A_E(\tau(u),\tau(v)).
\end{align*}
Define the conditional distribution $\mu^\sigma$ over $[q]^V$ by for any $\tau \in [q]^V$,
\begin{align}\label{eq:cond}
    \mu^\sigma(\tau) = \frac{w^\sigma(\tau)}{Z^\sigma}, \text{ where } Z^\sigma = \sum_{\tau \in [q]^V}w^\sigma(\tau).
\end{align}
The above definition works for all partial configurations $\sigma$.
In particular, if $\sigma$ is a feasible partial configuration, then $\mu^\sigma$ is the distribution $\mu$ conditional on $\sigma$. 
Let $\mu^\sigma_S$ denote the marginal distribution on $S$ conditional on the partial configuration $\sigma$.

As mentioned before, when $A_E=J-I$, where $J$ is the all-$1$ matrix and $I$ is the identity matrix, and $A_v=\*1$ is the all-$1$ vector,
$\mu$ is uniform over proper $q$-colourings,
and $Z$ is the number of them.
For technical purposes, we also need to consider list colourings later,
where each $v$ may have a list $L_v$ of available colours, and no edge can be monochromatic.
This can also be modelled as a spin system, with $q$ being the total number of possible colours of all vertices, and $A_v$ encoding what colours are available for $v$.
Let $\mu$ be the uniform distribution over all proper colourings.
Our list colouring instances in fact come from $\mu^\sigma$ for some partial configuration $\sigma$.
This is because fixing a colour $c$ at a vertex $v$ is equivalent to removing the vertex $v$ and removing $c$ from the lists of all neighbours of $v$.
Note that if $q-\Delta\ge k$ for some $k$,
then for any list colouring instances obtained this way, for any $v$, $\abs{L_v}-\deg_G(v)\ge k$ as well.

More generally, in this paper, we consider the permissive spin systems so that all $\mu^\sigma$ are well-defined.

\begin{definition}\label{def:permissive}
  A spin system $\mu$ is \emph{permissive} if for any $\Lambda \subseteq V$, any $\sigma \in [q]^\Lambda$, $Z^\sigma > 0$.
\end{definition}
\begin{remark}
  All spin systems with soft constraints ($A_E(i,j)>0$ and $A_V(i) > 0$ for all $i,j \in[q]$) are permissive.  
  Many natural spin systems with hard constraints are also permissive. 
  For example, the hardcore model and any list colouring instance $(G,L)$ such that $|L_v| \geq \deg_G(v) + 1$ for any $v \in V$ are permissive.
\end{remark}

Note that in \Cref{def:permissive}, $\sigma$ is allowed to be infeasible,
as $w^\sigma$ does not consider the weight contributed from inside $\Lambda$.
This allows $\mu^{\sigma}$ to be well-defined, even for infeasible $\sigma$.
Our definition of coupling independence, \Cref{def:CI}, indeed allows infeasible partial configurations.
Essentially, for any (feasible or infeasible) $\sigma$ on $\Lambda\subset V$, only the values on the boundary of $\Lambda$ matters for $\mu^{\sigma}$,
and the values inside $\Lambda$ do not.

On the other hand, for any partial configuration $\sigma$ on $\Lambda$, if $\sigma$ is locally feasible, then $\sigma$ is also globally feasible. Formally,
\begin{align}\label{eq:local2global}
  \prod_{v \in \Lambda}A_V(\sigma(v))\prod_{\{u,v\}\in E: u\in \Lambda \land v \in \Lambda} A_E(\sigma(u),\sigma(v)) > 0 \quad \implies \quad \mu_\Lambda(\sigma) > 0.
\end{align}
We use $\-{supp}(\mu^\sigma_S)$ to denote the support of $\mu^\sigma_S$. Formally,
\begin{align*}
  \-{supp}(\mu^\sigma_S) := \set{\tau \in [q]^S \mid \mu^\sigma_S(\tau) > 0}.
\end{align*}
Note that by~\eqref{eq:local2global}, the set $\-{supp}(\mu^\sigma_S)$ is easy to compute as we only need to consider local assignments.

Let us also observe that permissive systems always have a marginal lower bound.

\begin{observation}\label{obs:permissive-lb}
  Let $q$, $\Delta$, $A_E$, and $A_V$ be constants.
  Then there is a constant $b=b_{q,\Delta,A_E,A_V}$ such that for any $G$ with maximum degree $\Delta$, any permissive spin system $(G,q,A_E,A_V)$ is $b$-marginally bounded.
\end{observation}
\begin{proof}
  Given a partial configuration $\sigma$ on $\Lambda$ and a vertex $v\in V\setminus \Lambda$,
  let $S$ be the set of $2$-hop neighbours of $v$ that is not in $\Lambda$.
  For any $c\in\-{supp}(\mu^\sigma_v)$, we have that
  \begin{align*}
    \mu^{\sigma}_v(c) = \sum_{\tau\in \-{supp}(\mu^{\sigma}_S)} \mu^{\sigma}_S(\tau) \mu^{\sigma\cup\tau}_v(c).
  \end{align*}
  For any $c\in\-{supp}(\mu^\sigma_v)$ and $\tau\in \-{supp}(\mu^{\sigma}_S)$, 
  as $S$ is not adjacent to $v$ and the system is permissive,
  $\mu^{\sigma\cup\tau}_v(c)>0$.
%  As long as $\tau\in \-{supp}(\mu^{\sigma}_S)$, $\sigma\cup\tau$ is locally-feasible.
%  Since $S$ is not adjacent to $v$ and $c\in\-{supp}(\mu^\sigma_v)$,
%  we have that $\sigma\cup\tau\cup\{v\gets c\}$ is locally-feasible.
%  By \eqref{eq:local2global}, $\mu^{\sigma\cup\tau}_v(c)>0$.
  As both $\-{supp}(\mu^{\sigma}_S)$ and $\-{supp}(\mu^\sigma_v)$ are finite,
  there is a minimum of $\mu^{\sigma\cup\tau}_v(c)$, say $b$, over all choices of $\tau\in\-{supp}(\mu^{\sigma}_S)$ and $c\in \-{supp}(\mu^\sigma_v)$.
  Note that this $b$ may depend on $q$, $\Delta$, $A_E$, and $A_V$.
  It implies that
  \begin{align*}
    \mu^{\sigma}_v(c) &\ge b\sum_{\tau\in \-{supp}(\mu^{\sigma}_S)} \mu^{\sigma}_S(\tau) = b.\qedhere
  \end{align*}
\end{proof}

\section{Total influence decay}
\label{sec:influence-decay}

Given a graph $G=(V,E)$,
let $B_\ell(v)$ denote the ball of radius $\ell$ centered at $v$.
Namely, $B_{\ell}(v)\defeq\{u\mid u\in V,~\dist_G(u,v)\le \ell\}$,
where the distance is graph distance in $G$.
Let $S_\ell(v)$ denote the sphere of radius $\ell$ centered at $v$,
namely, $S_{\ell}(v)\defeq\{u\mid u\in V,~\dist_G(u,v)= \ell\}$.
In other words, $S_{\ell+1}(v)=\partial B_{\ell}(v)$, where $\partial B_\ell (v)$ denote the out-boundary of $B_\ell(v)$.

\begin{definition}[Total influence decay] \label{def:TI-decay}
Let $\delta: \mathbb{N} \to \mathbb{R}$ be a non-increasing function. We say  a Gibbs distribution $\mu$ satisfies \emph{total influence decay} with rate $\delta$ if for any two partial configurations $\sigma$ and $\tau$ on $\Lambda\subseteq V$ such that they disagree only on some $v \in \Lambda$, for any integer $\ell > 0$,
\begin{align}\label{eqn:TI-decay}
  \sum_{u \in S_{\ell}(v)}\dTV(\mu_u^\sigma,\mu_u^\tau) \leq \delta(\ell).
\end{align}
\end{definition}

The influence defined in~\eqref{eqn:TI-decay} was often used in recent works for spectral independence~\cite{ALO20,FGYZ22,CGSV21}.
The definition in~\eqref{eqn:TI-decay} should be distinguished from the definition of Dobrushin's influence in~\eqref{eqn:defDob}.
The following theorem shows that the coupling independence implies total influence decay.

\begin{theorem}  \label{thm:CI->sphere-CI}
  Suppose $\mu$ satisfies \emph{$C$-coupling independence}.
  Let $R>0$ be an integer.
  For any two partial configurations $\sigma$ and $\tau$ on $\Lambda\subseteq V$ such that they disagree only on some $v\in\Lambda$,
  there is a coupling $\+C$ between $\mu^{\sigma}$ and $\mu^{\tau}$ such that
  \begin{align*}
    \Ex_{(\sigma',\tau')\sim\+C}\left[ \sum_{u\in S_{R}(v)} \*{1}_{\sigma'(u)\neq\tau'(u)} \right] \le 2C\left(\frac{1}{2}\right)^{\ceil{\frac{R}{2C}}}.
  \end{align*}  
  As a consequence,  $\mu$ satisfies total influence decay with rate $\delta(x) = 2C\cdot 2^{-\ceil{\frac{x}{2C}}}$.
\end{theorem}

\begin{proof}
  Let $T \defeq \ceil{\frac{R}{2C}}$. 
  %Fix two partial pinnings $\sigma$ and $\tau$ on $\Lambda$ such that they disagree only at $v \in \Lambda$. 
  We construct $\+C$ via a $T$-step coupling procedure as follows. The whole procedure will generate a random pair $X,Y \in [q]^V$ such that $X \sim \mu^\sigma$ and $Y \sim \mu^\tau$.

\textbf{Step 1.} Since $\mu$ satisfies $C$-coupling independence, there exists a coupling $\+C_1$ of $\mu^\sigma$ and $\mu^\tau$ such that $\sum_{w \in V}\Pr_{(\sigma',\tau')\sim \+C_1}[\sigma'(w) \neq \tau'(w)] \leq C$. By an averaging argument, there exists an integer $1 \leq \ell(1) \leq 2C$ such that the expected disagreement on the sphere $S_{\ell(1)}(v)$ is at most $\frac{C}{2C} = \frac{1}{2}$, i.e.,
  \[\sum_{w \in S_{\ell(1)}(v)}\Pr_{(\sigma',\tau')\sim \+C_1}[\sigma'(w) \neq \tau'(w)] \leq \frac{1}{2}.
  \]
If there are multiple choices of such $\ell(1)$, pick the smallest one.
Denote $B_1 = B_{\ell(1)}(v)$ and $S_1 = S_{\ell(1)}(v)$.
We use the coupling $\+C_1$ to draw a pair of random configurations and project them to the set $D = B_1 \setminus \Lambda$. We then extend $\sigma,\tau$ on $\Lambda$ to a pair of random configurations $X,Y$ on $\Lambda \cup B_1$. Formally, $X(\Lambda) = \sigma$, $Y(\Lambda) = \tau$, and
\begin{align}\label{eq:ex-1}
  X(D) \sim \mu^{\sigma}_{D}, Y(D) \sim \mu^{\tau}_{D}, \text{ and } \Ex[ \dist\{X(S_1),Y(S_1)\}]\leq \frac{1}{2}.  
\end{align}
 %Note that $X(B_1)$ must be consistent with $\sigma$ on $B_1 \cap \Lambda$. The similar result holds for $Y(B_1)$.

\textbf{Step $k$ for $1 < k < T$.} Suppose we have obtained a pair of partial configurations $X$ and $Y$ on $B_{k-1} \cup \Lambda$, where $B_{k-1} = B_{\ell(k-1)}(v)$ and $\ell(k-1) \geq 1$ is an integer. 
%Assume all $X,Y,\ell(k-1)$ are fixed.
Next we extend the two configurations $X,Y$ onto a larger set $B_k \cup \Lambda$, where $B_k = B_{\ell(k)}(v)$ for some $\ell(k) > \ell(k-1)$. 
%Consider the coupling of two distributions $\mu^{\sigma \cup X(B_{k-1})}$ and $\mu^{\tau \cup Y(B_{k-1})}$, where $\sigma \cup X(B_{k-1})$ denotes a pinning on set $\Lambda \cup B_{k-1}$ (note that $\sigma$ and $B_{k-1}$ must be consistent on $\Lambda \cap B_{k-1}$) and $\tau \cup Y(B_{k-1})$ is defined similarly. 
Let $H = V \setminus B_{k-1}$ be the set of vertices outside the ball $B_{k-1}$.  
Let $\Lambda_H = \Lambda \cap H$.
By conditional independence, $\mu^X_H = \mu^{X(\Lambda_H\cup S_{k-1})}_H$ and $\mu^Y_H = \mu^{Y(\Lambda_H\cup S_{k-1})}_H$.
The two pinnings $X(\Lambda_H\cup S_{k-1})$ and $Y(\Lambda_H\cup S_{k-1})$ can disagree only at the set $S_{k-1}=S_{\ell(k-1)}(v)$, the sphere of radius $\ell(k-1)$ centered at $v$. 
Suppose the disagreements of $X(S_{k-1})$ and $Y(S_{k-1})$ can be listed as $v_1,v_2,\ldots,v_m \in S_{k-1}$, where $m = \dist\{X(S_{k-1}),Y(S_{k-1})\}$. We can define a sequence of pinnings $\sigma_0,\sigma_1,\ldots,\sigma_m$ on $S_{k-1}\cup \Lambda_H$ such that $\sigma_0 = X(\Lambda_H \cup S_{k-1})$, $\sigma_i$ is obtained from $\sigma_{i-1}$ by changing the value at $v_i$ from $X(v_i)$ to $Y(v_i)$, and so $\sigma_m = Y(\Lambda_H \cup S_{k-1}) $. By the coupling independence, for any $i \in [m]$, one can couple $\mu^{\sigma_i}$ and $\mu^{\sigma_{i-1}}$ such that the expected hamming distance is at most $C$.\footnote{Two pinnings $\sigma_i,\sigma_{i-1}$ can be improper partial configurations but conditional distributions $\mu^{\sigma_i}$ and $\mu^{\sigma_{i-1}}$ are defined in~\eqref{eq:cond}. The coupling independence condition holds for all possible pinnings including infeasible ones.} 
By the triangle inequality of the Wasserstein distance, 
we can couple $\mu^{\sigma_0}$ and $\mu^{\sigma_m}$ so that the expected Hamming distance is at most $mC$.
By projecting this coupling into the subset $H$, we obtain a coupling $\+C_k$ of $\mu_H^X = \mu_{H}^{\sigma_0}$ and $\mu_H^Y = \mu_{H}^{\sigma_m}$ such that $\Ex_{(\sigma_k,\tau_k)\sim \+C_k}[\dist\{ \sigma_k,\tau_k\}]\leq mC$. 
Again, by an averaging argument, there exists $\ell(k)$ such that $\ell(k-1) + 1 \leq \ell(k) \leq \ell(k-1) + 2C$ and 
\begin{align*}
  \sum_{w \in S_{\ell(k)}(v)}\Pr_{(\sigma_k,\tau_k)\sim \+C_k}[\sigma_k(w) \neq \tau_k(w)] \leq \frac{mC}{2C} = \frac{m}{2}=\frac{\dist\{X(S_{k-1}),Y(S_{k-1})\}}{2}. 
\end{align*}
If there are multiple choices of such $\ell(k)$, pick the smallest one.
Let $B_k = B_{\ell(k)}(v)$ and $S_k = S_{\ell(k)}(v)$. Similar to Step 1, we use the coupling $\+C_k$ to draw a pair of configurations $(\sigma_k,\tau_k)$. 
Let $D_k = B_k \setminus (B_{k-1} \cup \Lambda)$. We further extend $X,Y$ to the set $D_k$ by setting $X(D_k) = \sigma_k(D_k)$ and $Y(D_k) = \tau_k(D_k)$. We have 
\[X(D_k)\sim \mu^{X(B_{k-1}\cup \Lambda)}_{D_k},   Y(D_k)\sim \mu^{Y(B_{k-1}\cup \Lambda)}_{D_k}.\]
In other words, $X$ and $Y$ are now partial configurations on set $B_{k-1} \cup \Lambda \cup D_k = B_{k} \cup \Lambda$. It holds that
\begin{align}\label{eq:ex-2}
\Ex[ \dist\{X(S_k),Y(S_k)\} \mid X(B_{k-1} \cup \Lambda), Y(B_{k-1} \cup \Lambda),\ell(k-1)]\leq \frac{\dist\{X(S_{k-1}),Y(S_{k-1})\}}{2}. 
\end{align}

\textbf{Step $T$.} The last step is similar to the general step $k$.
Assume $X,Y$ are two fixed configurations on $B_{T-1} \cup \Lambda$, where $B_{T-1} = B_{\ell(T-1)}(v)$ and $\ell(T-1)$ is a fixed integer.
Let $H = V \setminus B_{T-1}$.
Let $\+C_T$ be the coupling obtained in the same way as the general step from the triangle inequality.
We use the coupling $\+C_T$ to sample a pair of configurations $(\sigma_T,\tau_T)$.
We then further extend $X,Y$ to the set $D_T = V \setminus (B_{T-1} \cup \Lambda)$ by setting $X(D_T) = \sigma'(D_T)$ and $Y(D_T) = \tau'(D_T)$. We have
\begin{align*}
  X(D_T)\sim \mu^{X(B_{T-1}\cup \Lambda)}_{D_T},   Y(D_T)\sim \mu^{Y(B_{T-1}\cup \Lambda)}_{D_T}. 
\end{align*}
Note that $\ell(T - 1) \leq 2C(T-1) = 2C(\ceil{\frac{R}{2C}}-1) < R$, so that $S_R(v) \subseteq H$. It holds that
\begin{align}\label{eq:ex-3}
  &\Ex[ \dist\{X(S_R(v)),Y(S_R(v))\} \mid X(B_{T-1} \cup \Lambda), Y(B_{T-1} \cup \Lambda),\ell(T-1)]\notag\\
  \leq\,& C \cdot {\dist\{X(S_{T-1}),Y(S_{T-1})\}}.
\end{align}
The above inequality holds because the expected disagreement on $S_R(v)$ is at most the total expected disagreement produced by $\+C_T$, which is at most  $C\cdot{\dist\{X(S_{T-1}),Y(S_{T-1})\}}$. %due to coupling independence and path coupling.

We now show that the above procedure is a valid coupling between $\mu^\sigma$ and $\mu^\tau$. The procedure can be viewed as follows. Initially, $X$ and $Y$ are partial configurations on $D_0 = \Lambda$ with $X = \sigma$ and $Y = \tau$. In the $k$-th step, we extend $X,Y$ to a new set $D_k$. Hence $X,Y$ are partial configurations on $\cup_{i \leq k}D_i$ after the $k$-th step. The tricky part here is that all $D_1,D_2,\ldots,D_T$ are random variables. The following property is the key to proving the correctness of the coupling:
for every step $1 \leq k \leq T$, the coupling satisfies
\begin{itemize}
  \item given $(D_i)_{i < k}$, $X(D_{0:k-1})$, and $Y(D_{0:k-1})$, the set $D_k$ is fixed, where $D_{0:k-1}=\cup_{0\leq i \leq k-1}D_i$;
  \item $X(D_k) \sim \mu_{D_k}^{X(D_{0:k-1})}$ and $Y(D_k) \sim \mu_{D_k}^{Y(D_{0:k-1})}$.
\end{itemize}
This property is easy to verify from the construction of $X$ and $Y$.

Next we show $X \sim \mu^\sigma$. A similar proof shows that $Y \sim \mu^\tau$.
In every step, we sample $X(D_k)$ and $Y(D_k)$ jointly from a coupling. This can be viewed as a two-step process.
We first sample $X(D_k)$, and given the value of $X(D_k)$, the coupling specifies a conditional distribution of $Y(D_k)$. 
Then we sample $Y(D_k)$ from this conditional distribution using independent randomness $\+R_k$ (e.g., $\+R_k$ can be a uniform real number in $(0,1)$). We fix the randomness $\+R_1,\+R_2,\ldots,\+R_T$. For any configuration $\rho \in [q]^V$, we compute the probability that $X = \rho$. Initially, $D_0 = \Lambda$, $X(D_0)=\sigma$ and $Y(D_0) = \tau$. Hence, $X = \rho$ only if $\rho(\Lambda) =\sigma$. In the first step, by the property above, $D_1$ is fixed, and $X = \rho$ implies that $X(D_1) = \rho(D_1)$, which happens with probability $\mu^{\rho(D_0)}_{D_1}(\rho(D_1))$. Also, $Y(D_1)$ is fixed because $X(D_1) = \rho(D_1)$ and $\+R_1$ is fixed. By induction, in the $k$-th step, $D_k$ is fixed, $X(D_k)=\rho(D_k)$ with probability $\mu^{\rho(D_{0:k-1})}_{D_k}(\rho(D_k))$, and $Y(D_k)$ is fixed by $X(D_k)$ and $\+R_k$. By the chain rule,  
\begin{align*}
  \Pr[X = \rho \mid \+R_1,\+R_2,\ldots,\+R_T] = \*{1}_{\rho(\Lambda) =\sigma} \cdot \prod_{k=1}^T\mu^{\rho(D_{0:k-1})}_{D_k}(\rho(D_k)) = \mu^\sigma(\rho).
\end{align*}
Taking expectation over $\+R_1,\+R_2,\ldots,\+R_T$ in both sides shows that $X \sim \mu^\sigma$. 

Finally, using ~\eqref{eq:ex-1}, \eqref{eq:ex-2}, and \eqref{eq:ex-3}, the expected disagreement on $S_R(v)$ can be bounded as
\begin{align*}
  \Ex[ \dist\{X(S_R(v)),Y(S_R(v))\}] &\leq C  \Ex[\dist\{X(S_{T-1}),Y(S_{T-1})\}] \\
  &\leq \frac{C}{2^{T-2}} \Ex[\dist\{X(S_{1}),Y(S_{1})\}]\\
  &\leq \frac{C}{2^{T-1}}.
\end{align*}
The total influence decay consequence follows from the coupling inequality.
%
%Suppose $X_{k-1}$ and $Y_{k-1}$ disagree at the set $S' \subseteq S_{k-1} = S_{\ell(k-1)}(v)$. 
\end{proof}

\section{A recursive marginal estimator}\label{sec:alg}

%\begin{definition}[Total influence decay]
%    The Gibbs distribution $\mu$ on $G=(V,E)$ exhibits total influence decay with constant $C > 0$ and decay rate $\kappa > 0$ if for any proper partial pinnings $\tau$ supported on $\Lambda$, vertex $u \in V \setminus \Lambda$ and positive integer $R$, it holds that
%    \begin{align}\label{eqn:TI-decay}
%        \forall c_1,c_2 \in \+A_u,\quad \sum_{v \in \partial B(u,R)} \DTV{\mu^{\tau, u \gets c_1}_v}{\mu^{\tau, u \gets c_2}_v} \le C (1-\kappa)^R.
%    \end{align}
%\end{definition}

%The main goal of this section is to establish the following theorem, which gives deterministic algorithm to estimate the marginal probability of a Gibbs distribution under the assumption of total influence decay.
We show \Cref{thm:main} in this section.
As in the condition of \Cref{thm:main}, throughout the section, we assume that the underlying graphs $G$ for the Gibbs distributions have constant maximum degree~$\Delta$.
Moreover, we assume a marginal lower bound $0<b<1$ as in \Cref{def:marginal-bound}, and exponential total influence decay, namely the bound in \eqref{eqn:TI-decay} with $\delta(\ell)=\exp(-\Omega(\ell))$. 

Our algorithm requires a parameter $R$, the radius at which we couple vertices.
For an integer $k\ge 0$, let $H(k)$ be the harmonic sum defined by $H(k) \defeq \sum_{i=1}^k \frac{1}{i}$ with the convention that $H(0)=0$.
We choose a sufficiently large integer $R$ such that
\begin{align} \label{eqn:select-R}
  30\delta(R) H(\Delta^R) < b^4.
\end{align}
By \Cref{thm:CI->sphere-CI},
we can take $\delta(R)=2C 2^{-\ceil{\frac{R}{2C}}}$,
where $C$ is the coupling independence constant.
Since $H(\Delta^R) = O(R\log\Delta)$, some $R=O(C(\log b^{-1}+\log C+\log\log\Delta))$ suffices.

In \Cref{sec:rec-marginal-resolver}, we introduce an LP based algorithm that takes as inputs estimated ratios of marginal probabilities of some partial configurations,
and outputs an estimation of the marginal ratio of a particular vertex, with a better approximation guarantee.
The error analysis of this algorithm is given in \Cref{sec:analysis}.
These input ratios can be written as the product of two ratios, each of which is regarding a single vertex.
Thus, we have a recursive algorithm to estimate the marginal probability, described in \Cref{sec:determistic-algo}.

\newcommand\paradelta{\ensuremath{b^{-1}\delta(R)}}
\subsection{Marginal estimation via linear programming} \label{sec:rec-marginal-resolver}
Suppose we want to estimate the marginal probability of some vertex $u$.
It is equivalent to approximate the ratio between marginals where any two different values are assigned to $u$.
The basic building block of our algorithm is an estimator which takes estimations of marginal ratios of certain partial configurations,
and outputs an estimation of the marginal ratio for $u$ with a better approximation guarantee.
We construct a linear program, similar to the one used by Moitra \cite{Moi19}, to certify a coupling between assigning $u$ to two different values.
However, the ``coupling'' we choose is different, and is inspired by the one used by Chen, Liu, Mani, and Moitra \cite{CLMM23}.
For readers not familiar with Moitra's approach,
we provide some heuristics and intuition behind it in \Cref{sec:heuristics}.

Given two partial configurations $\sigma$ and $\tau$ on $\Lambda\subseteq V$ which differ at only one vertex $u\in \Lambda$,
we want to couple $\mu^\sigma$ and $\mu^\tau$ by coupling vertices in $S_R(u)=\partial B_R(u)$ for some radius $R>0$.
We choose a vertex $v$ from $S_R(u)$ uniformly at random and couple it optimally (in the sense of TV distance) between its two marginal distributions.
If the coupling failed, we immediately stop the whole process, and otherwise we continue to couple the next randomly chosen vertex.
The process also ends when all of $S_R(u)$ has been considered.
One may have noticed that this is not a complete coupling, but rather a partial one.
Formally it is described in \Cref{alg:coupling}.
For $\sigma\in[q]^\Lambda$ and $v\not\in \Lambda$, the notation $\sigma_{v\gets c}$ means a partial configuration which agrees with $\sigma$ on $\Lambda$ and assigns $c$ to $v$.

\begin{algorithm}
    \caption{The partial coupling}
    \label{alg:coupling}    
    \SetKwInOut{Input}{Input}
    \SetKwInOut{Output}{Output}
    \underline{Partial-Coupling} $(\sigma,\tau)$\;
    \Input{Partial configurations $\sigma,\tau \in [q]^\Lambda$ that only differ at a vertex $u \in \Lambda$}
    \Output{A pair of partial configurations $\sigma'$ and $\tau'$ over some $\Lambda'\supseteq \Lambda$}
    $\sigma'\gets\sigma$, $\tau'\gets\tau$\;
    \While{$S_R(u)\setminus \Lambda\neq \emptyset$}{
      Choose $v$ from $S_R(u)\setminus \Lambda$ uniformly at random\;
      Draw $(c_1,c_2)$ from the optimal coupling between $\mu_v^{\sigma'}$ and $\mu_v^{\tau'}$\label{line:TV-coupling}\;
      $\sigma'\gets\sigma'_{v\gets c_1}$, $\tau'\gets\tau'_{v\gets c_2}$\;
      $\Lambda \gets \Lambda \cup \{v\}$\;
      \If{$c_1\neq c_2$}
      {\Return $(\sigma',\tau')$\;}
    }
    \Return $(\sigma',\tau')$\;
\end{algorithm}

Note that our algorithm does not really need to construct the partial coupling described in \Cref{alg:coupling}.
Instead, we construct linear programs that mimic the coupling process.
We start by defining the coupling tree $\+T$ with root $rt$. 
This tree essentially enumerates all possible intermediate states of \Cref{alg:coupling}.
Formally the coupling tree is constructed by \Cref{alg:coupling-tree}.
%The sets ${\+A}^{\sigma}_{rt,v}$ and ${\+A}^{\tau}_{rt,v}$ are the sets of feasible values of $v$ under $\sigma$ and $\tau$, respectively.

\begin{algorithm}
    \caption{Construction of the coupling tree}
    \label{alg:coupling-tree}
    \SetKwInOut{Input}{Input}
    \SetKwInOut{Output}{Output}
    \SetKwInOut{Parameter}{Parameter}
    \underline{Coupling-Tree} $(\sigma,\tau)$\;
    \Input{Partial configurations $\sigma,\tau \in [q]^\Lambda$ that only differ at a vertex $u \in \Lambda$}
    \Output{A coupling tree $\+T$ with its root $rt$}
    \Parameter{A positive integer $R$}
    Construct a tree $\+T$ containing a single root node $rt$ with $\mathrm{label}(rt) \gets \tp{\sigma,\tau,\Lambda,\emptyset}$ \; % and $\mathrm{\ell}(rt) \gets \abs{\partial B_R(u) \setminus S}$\;
    \For{$v \in S_R(u)\setminus \Lambda$}{
        %${\+A}^{\sigma}_{rt,v} \gets \{c \in Q_v \mid \mu_{S \cup \{v\}}(\sigma_{v \gets c}) > 0\}$ and ${\+A}^{\tau}_{rt,v} \gets \{c \in Q_v \mid \mu_{S \cup \{v\}}(\tau_{v \gets c}) > 0\}$\;
        %\For{$(c_1,c_2) \in {\+A}^\sigma_{rt,v} \times {\+A}^{\tau}_{rt,v}$}{
        \For{$(c_1,c_2) \in \mathrm{supp}(\mu^\sigma_v) \times \mathrm{supp}(\mu^\tau_v)$}{
            \eIf{$c_1=c_2=c$}{
                $(\+T_{v,c},rt_{v,c}) \gets$ Coupling-Tree $(\sigma_{v \gets c}, \tau_{v \gets c})$\;
                Append $\+T_{v,c}$ to $rt$ by connecting $rt$ and $rt_{v,c}$\;
            }
            {
                Introduce a leaf node $w$ with $\mathrm{label}(w) \gets (\sigma_{v \gets c_1}, \tau_{v \gets c_2}, \Lambda \cup \{v\},v)$\;
                Connect $rt$ and $w$\;
            }
        }
    }
    \Return $\tp{\+T,rt}$\;
\end{algorithm}

Denote by $\+V(\+T)$ the set of nodes of $\+T$ and $\+L(\+T)$ the set of leaves.
Each node $w$ of the tree represents an intermediate state of \Cref{alg:coupling}.
Its label $(\sigma,\tau,\Lambda,D)$ represents the two partial configurations $\sigma$ and $\tau$, the pinned vertex set $\Lambda$, and the set $D$ of differing vertices other than the initial disagreement $u$.
In fact, $D$ can only be either some vertex $v$ or the empty set $\emptyset$, because \Cref{alg:coupling} stops whenever the first disagreement other than $u$ is introduced.
Thus, if $D\neq\emptyset$, then the node must be a leaf.
We call a leaf node $\emph{good}$ if its $D=\emptyset$.
Namely, denote by $\+{GL}(\+T) = \{w \in \+L(\+T) \mid \mathrm{label}(w) = (\ast,\ast,\ast,\emptyset)\}$ the set of good leaves,
where $\ast$ denotes any possible value at that position.
The rest of the leaves $\+{BL}(\+T) = \+L(\+T) \setminus \+{GL}(\+T)$ are \emph{bad}.
Moreover, denote by $\+C(w)$ the set of children of a node~$w$.

The linear program is introduced in \Cref{alg:LP}.
Again, some intuition and heuristics for this approach are given in \Cref{sec:heuristics}.
For each $w\in\+V(\+T)$ whose label is $(\sigma^w,\tau^w,\Lambda^w,\ast)$, we call $r_w \defeq \frac{\mu_{\Lambda^w}(\sigma^w)}{\mu_{\Lambda^w}(\tau^w)}$ its marginal ratio.
Our goal is to estimate the marginal ratio $r = r_{rt}= \frac{\mu_\Lambda(\sigma)}{\mu_\Lambda(\tau)}$ for the root of the coupling tree.
We do so by combining the LP with a binary search.
For some guessed upper and lower bounds $r_+$ and $r_-$ for $r$,
ideally, we want to construst an LP such that it is feasible if and only if $r_-\le r\le r_+$.
Because of the presence of errors, eventually, we will only establish an approximation version of this claim.

%This linear program gives an implicit function that takes the marginal ratios for leaves in $\+L(\+T)$ as inputs and outputs the marginal ratio of ${rt}$.
%When marginal ratios $\*R$ for all the leaves $w \in \+L(\+T)$ are already provided with their accurate value.
%Then, take $\eps = 0$, one can show that the linear program in \Cref{alg:LP} has feasible solution if and only if $r_- \leq r \leq r_+$, where $r$ is the accurate value of the marginal ratio of $\-{rt}$.
%
The LP contains two variables $x_w$ and $y_w$ for each node $w\in\+V(\+T)$ of the coupling tree.
Intuitively, $x_w$ represents $\frac{z_w \mu_\Lambda(\sigma)}{\mu_{\Lambda^w}(\sigma^w)}$ and $y_w$ represents $\frac{z_w \mu_\Lambda(\tau)}{\mu_{\Lambda^w}(\tau^w)}$, 
where $z_w$ is the probability that the coupling reaches the node $w$. 
The four types of constraints in \Cref{alg:LP} can be intuitively interpreted when $x_w=\frac{z_w \mu_\Lambda(\sigma)}{\mu_{\Lambda^w}(\sigma^w)}$ and $y_w =\frac{z_w \mu_\Lambda(\tau)}{\mu_{\Lambda^w}(\tau^w)}$ as follows.
We only explain the intuitions for $x_w$, and the same applies to $y_w$.
%as follows, where we explain the intuition for constraints on $x_w$, and the same intuition applies to constraints on $y_w$.
\begin{enumerate}
  \item Validity constraints: When $w = rt$ is the root, we have that $z_w = 1$ and $x_w = \frac{\mu_\Lambda(\sigma)}{\mu_\Lambda(\sigma)} = 1$.
  \item Recursive constraints: Given a node $w$ with label $(\sigma^w,\tau^w,\Lambda^w,\emptyset)$, fix a vertex $v \in S_R(u)\setminus \Lambda^w$ and a value $c \in \text{supp}(\mu_v^{\sigma^w})$. We verify the following identity from the coupling process
%the coupling picks a uniformly random vertex $v \in S_R(u)\setminus \Lambda^w$ and samples the values of $v$ from the optimal coupling between $\mu_v^{\sigma^w}$ and $\mu_v^{\tau^w}$. Fix $v$ and $c \in \text{supp}(\mu_v^{\sigma^w})$, if $x_w = \frac{z_w \mu_\Lambda(\sigma)}{\mu_{\Lambda^w}(\sigma^w)}$, then 
  \begin{align}\label{eq:recursive-constraint-intuition-0}
    \sum_{\substack{w' \in \+C(w):\\\mathrm{label}(w')=\tp{\sigma^w_{v \gets c},\ast, \ast, \ast} }}\frac{x_{w'}}{x_w} = \frac{1}{\mu^{\sigma^w}_v(c)} \sum_{\substack{w' \in \+C(w):\\\mathrm{label}(w')=\tp{\sigma^w_{v \gets c},\ast, \ast, \ast} }}\frac{z_{w'}}{z_w} = \frac{1}{\ell(w)},
  \end{align} 
  where $\ell(w)\defeq\abs{S_R(u) \setminus \Lambda^w}$ for a node $w\in\+V(\+T)$ with label $(\ast,\ast,\Lambda^w,\ast)$.
   The first equality holds because $x_w = \frac{z_w \mu_\Lambda(\sigma)}{\mu_{\Lambda^w}(\sigma^w)}$.
   The second equality holds because $\frac{z_{w'}}{z_w}$ is the probability of the coupling reaches $w'$ from $w$.
   Thus, the summation of $\frac{z_{w'}}{z_w}$ is the probability that the vertex $v$ is chosen in this step of the coupling, and is given the value $c$ on $\sigma$'s side.
   This happens with probability $\frac{\mu^{\sigma^w}_v(c)}{\ell(w)}$.
  \item Leaf constraints: The ratio $\frac{x_w}{y_w} = \frac{\mu_\Lambda(\sigma)}{\mu_\Lambda(\tau)} \cdot \frac{\mu_{\Lambda^w}(\tau^w)}{\mu_{\Lambda^w}(\sigma^w)}$ is a product of two ratios. Assume  $R_w \approx \frac{\mu_{\Lambda^w}(\sigma^w)}{\mu_{\Lambda^w}(\tau^w)}$ and $r_- \le \frac{\mu_\Lambda(\sigma)}{\mu_\Lambda(\tau)} \le r_+$. We have  $r_-R_w^{-1} \lesssim \frac{x_w}{y_w} \lesssim r_+R_w^{-1}$.
    Note that for good leaves, the ratio $R_w$ can be efficiently computed and the constraint contains no error.
    However, for bad leaves, we have to settle on an approximate version.
   \item Overflow constraints:
  These control the probability of going to a bad leaf:
  \begin{align*}
    \sum_{w' \in \+C(w) \cap \+{BL}(\+T)} \frac{x_{w'}}{x_w} =  \frac{1}{\mu^{\sigma^w}_v(c)} \sum_{w' \in \+C(w) \cap \+{BL}(\+T)}\frac{z_{w'}}{z_w} \leq \frac{\Pr[\text{\,coupling goes from $w$ to a bad leaf}\,]}{b},
  \end{align*}
  where the last inequality follows from the marginal lower bound $\mu_v^{\sigma^w}(c) \geq b$ and the probability that $w$ goes to a bad leaf is at most $\frac{\delta(R)}{\ell(w)}$ due to the influence bound.
  Here we need to use $b^{-1}$ to upper bound $\mu_v^{\sigma^w}(c)^{-1}$ because the latter is a hard to compute quantity.
  Because of the inequality above, we set the parameter $\eta$ to
  \begin{align}\label{eq:eta}
    \eta \defeq ~\paradelta \quad \text{where $R$ is defined in \eqref{eqn:select-R}}.
  \end{align}
  %This is because the coupling picks the next vertex $v$ from $S_R(u)\setminus \Lambda^w$ uniformly at random with probability $\frac{1}{\ell(w)}$, and then the coupling on vertex $v$ fails with a low probability. Formally, the probability that coupling goes from $w$ to a bad leaf is 
  %the coupling picks the next vertex $v$ from $S_R(u)\setminus \Lambda^w$ uniformly at random with probability $\frac{1}{\ell(w)}$, and then the coupling on vertex $v$ fails with a low probability. Formally, the probability that coupling goes from $w$ to a bad leaf is 
  %\begin{align}\label{eq:overflow-constraint-intuition-2}
  %  \sum_{w' \in \+C(w) \cap \+{BL}(\+T)} \frac{z_{w'}}{z_w} = \sum_{v \in S_R(u)\setminus \Lambda^w} \frac{1}{\ell(w)} \DTV{\mu_v^{\sigma^w}}{\mu_v^{\tau^w}} \leq \frac{\delta(R)}{\ell(w)}.
  %\end{align}
  %Combining \eqref{eq:overflow-constraint-intuition-2} with \eqref{eq:recursive-constraint-intuition-1} and using the marginal lower bound, we have
  %\begin{align*}
  %  \sum_{w' \in \+C(w) \cap \+{BL}(\+T)} \frac{x_{w'}}{x_w} = \sum_{w' \in \+C(w) \cap \+{BL}(\+T)} \frac{z_w}{z_{w'}} \cdot \frac{1}{\mu^{\sigma^w}_v(c)} \leq \frac{\delta(R)}{\ell(w)}.
  %\end{align*}
\end{enumerate}

\Cref{lem:completeness} formally verifies that $x_w=\frac{z_w \mu_\Lambda(\sigma)}{\mu_{\Lambda^w}(\sigma^w)}$ and $y_w =\frac{z_w \mu_\Lambda(\tau)}{\mu_{\Lambda^w}(\tau^w)}$ do satisfy all the constraints.
%Furthermore, we also show that if the LP has a feasible solution, then the true ratio $r = \frac{\mu_\Lambda(\sigma)}{\mu_\Lambda(\tau)}$ satisfies $r_- \lesssim r \lesssim r_+$.
When applying Moitra's method~\cite{Moi19}, it is standard to choose these variables and constraints (1) and (2),
as well as the constraints for good leaves where there is no error.
On the other hand, typical applications of Moitra's method involve some local uniformity constraints to control the probability of reaching bad leaves, and put no constraints on the bad leaves themselves.
Local uniformity no longer holds in our setting.
Instead, we include the overflow constraints to bound the effect of bad leaves.
These constraints can only reduce bad leaves' effects on the overall error by some constant factor.\footnote{A sharp-eyed reader may have noticed that we can reduce the effect of bad leaves to polynomially small by setting the radius $R$ to $\Omega(\log n)$. This is correct, but doing so would increase the size of $S_R(u)$ to polynomially large, and the overall LP would be exponentially large.}
Thus, we also introduce the leaf constraints on the bad leaves.
Their marginal ratios are involved in these constraints, which we recursively solve.
Overall, the error in our algorithm decreases by a constant factor in each iteration of the recursive call.

In other words, to construct our LP, we need the coupling tree $\+T$,
as well as the marginal ratios for all the leaf nodes of~$\+T$.
The ratios for good leaves can be efficiently computed, 
and the ratios for bad leaves are recursively fed and are denoted $\*R$.
The LP is combined with a binary search to find an estimate to $r$,
and the other two inputs $r_-$ and $r_+$ are our current guesses of the upper and lower bounds of $r$,
which we will keep adjusting during the binary search.
When $\*R$ has no error, a solution of this LP is guaranteed to exist for $r_-=r_+$.
However, as our $\*R$ may contain errors, we can only rely on binary search to reduce the gap $r_+-r_-$ to an appropriate level,
rather than require $r_+=r_-$.

\begin{algorithm}[ht]
    \caption{The linear program}
    \label{alg:LP}
    \SetKwInOut{Input}{Input}
    \SetKwInOut{Output}{Output}
    \SetKwInOut{Parameter}{Parameter}
    \underline{LP} $(r_-,r_+,\+T,\*R, \eps)$\;
    \Input{Positive real values $r_- \le r_+$, a coupling tree $\+T$, marginal ratio estimates $\*R \in \mathbb{R}_{> 0}^{\+L(\+T)}$, and their error margin~$\eps$}
    \Output{A Boolean value, indicating whether the LP has a feasible solution}
    \Parameter{$\eta>0$}
    \Return true if and only if the following LP has a feasible solution
    \begin{enumerate}
      \item Validity constraints:
        \begin{align*}
          \forall w \in \+V(\+T), \quad x_w,y_w \ge 0\\
          x_{rt} = y_{rt} = 1
        \end{align*}
      \item Recursive constraints:
        
        For any non-leaf node $w$ with $\mathrm{label}(w)=(\sigma^w,\tau^w,\Lambda^w,\emptyset)$, and $v \in S_R(u) \setminus \Lambda^w$, 
        \begin{align*}
          \forall c \in \mathrm{supp}(\mu^{\sigma^w}_v),\quad  \sum_{\substack{w' \in \+C(w):\\\mathrm{label}(w')=\tp{\sigma^w_{v \gets c},\ast, \ast, \ast} }} x_{w'} = \frac{x_w}{\ell(w)}\\
          \forall c \in \mathrm{supp}(\mu^{\tau^w}_v),\quad  \sum_{\substack{w' \in \+C(w):\\\mathrm{label}(w')=\tp{\ast,\tau^w_{v \gets c}, \ast, \ast} }} y_{w'} = \frac{y_w}{\ell(w)}
        \end{align*}
      \item Leaf constraints:
        \begin{align*}
          \forall w \in \+{GL}(\+T), \quad & r_- R_w^{-1} y_{w} \le x_w \le r_+ R^{-1}_w y_w\\
          \forall w \in \+{BL}(\+T), \quad & r_- \tp{1+\eps}^{-1} R_w^{-1} y_w \le x_w \le r_+ \tp{1+\eps} R_w^{-1} y_w
        \end{align*} 
      \item Overflow constraints: let $\eta \defeq ~\paradelta$, where $R$ is defined in \eqref{eqn:select-R},
        \begin{align*}
          \forall w \in \+V(\+T) \setminus \+L(\+T), \quad \sum_{w' \in \+C(w) \cap \+{BL}(\+T)} x_{w'} \le \frac{\eta}{\ell(w)} x_w\\
          \forall w \in \+V(\+T) \setminus \+L(\+T), \quad \sum_{\substack{w' \in \+C(w) \cap \+{BL}(\+T) }} y_{w'} \le \frac{\eta}{\ell(w)} y_w
        \end{align*}
    \end{enumerate}
\end{algorithm}

With the coupling tree $\+T$ for $(\sigma, \tau)$ (in \Cref{alg:coupling-tree}) and the linear program (in \Cref{alg:LP}) in hand, 
we then estimate the marginal ratio $r = \frac{\mu_\Lambda(\sigma)}{\mu_\Lambda(\tau)}$ by the binary search mentioned above.
We formally state this binary search in \Cref{algo:marginal-resolver}.
%We assume the marginal ratio for the leaves $\+L(\+T)$ are already provided by some oracle with relative error at most $\epsilon$.

\begin{algorithm}
  \caption{Marginal estimation based on LP}
  \label{algo:marginal-resolver}
    \SetKwInOut{Input}{Input}
    \SetKwInOut{Output}{Output}
    \SetKwInOut{Parameters}{Parameters}
    \underline{Marginal-estimator} $(\+T,\*R, \eps)$ \\
    \Input{A coupling tree $\+T$ with root $\-{rt}$, marginal ratio estimates $\*R \in \mathbb{R}_{> 0}^{\+L(\+T)}$ and their error margin $\eps$}
    \Output{An estimate $\hat{r}$}
    \Parameters{An integer $R>0$ and a real number $\eta>0$}
    $r_{low} \gets b$ and $r_{upp} \gets b^{-1}$\;
    $\widehat{\eps} \gets \eta H(\Delta^R) \cdot \eps$\;
    \While{$r_{upp} > (1+\widehat{\eps})^2 r_{low}$}{
        Let $m \gets (r_{upp}+r_{low})/2$\;
        \If{both LP$(r_{low},m,\+T,\*R, \eps)$ and LP$(m,r_{upp},\+T,\*R, \eps)$ are true\label{line:early-exit}}{
            \Return $m$\;
        }
        \tcp{Assertion: otherwise either LP$(r_{low},m,\+T,\*R, \eps)$ or LP$(m,r_{upp},\+T,\*R, \eps)$ is true}
        \eIf{LP$(r_{low},m,\+T,\*R, \eps)$ is true}{
            $r_{upp} \gets m$\;
        }
        {
            $r_{low} \gets m$\;
        }
    }
    \Return $\hat{r} = \sqrt{r_{low}r_{upp}}$\;
\end{algorithm}

The following bound is the main guarantee of the vector $\*R$.

\begin{condition}[$\eps$-error bound]\label{cond:LP}
  Let $\eps > 0$ be a parameter.
  Let $\mu$ be a Gibbs distribution on $[q]^V$.
  Let $\+T$ be the coupling tree of $(\sigma, \tau)$,
  where $\sigma, \tau \in [q]^\Lambda$ are two partial configurations on $\Lambda\subseteq V$ that only differ at some vertex $u \in \Lambda$.
  Then, let $\*R \in \^R^{\+L(\+T)}_{\geq 0}$ be a vector defined on the leaves of $\+T$.
  %When invoked with recursion depth $k$,
  For any leaf node $w\in\+L(\+T)$ with $\-{label}(w)=(\sigma^{w},\tau^{w},\Lambda^{w},\ast)$,
  \begin{itemize}
    \item if $w\in\+{GL}(\+T)$, $R_w= \frac{\mu_{\Lambda^w}(\sigma^w)}{\mu_{\Lambda^w}(\tau^w)}$;
    \item if $w\in\+{BL}(\+T)$, it holds that
      \begin{align*}
        (1+\eps)^{-1} \le R_w \cdot \frac{\mu_{\Lambda^w}(\tau^w)}{\mu_{\Lambda^w}(\sigma^w)} \le (1+\eps).
      \end{align*}
  \end{itemize}  
\end{condition}

When \Cref{cond:LP} holds, we say $\*R$ satisfies the $\eps$-error bound.

Our marginal estimator (\Cref{algo:marginal-resolver}) takes $(\+T, \*R, \eps)$ as input and outputs an estimate $\widehat{r}$ of $r$.
The key property of this estimator is that if $\*R$ satisfies the $\eps$-error bound, 
then the error of $\widehat{r}$ shrinks by a constant factor.

\begin{lemma} \label{lem:error-bound}
  Let $(\+T, \*R, \eps)$ be the input of \Cref{algo:marginal-resolver} such that $\*R$ satisfies the $\eps$-error bound (namely \Cref{cond:LP} holds) for some $\eps\le 3b^{-2}$.
  %Let $\-{rt}$ be the root of $\+T$ and let $\ell = \ell(\-{rt})$ for simplicity.
%  Let $R, \eta$ be the parameters used in \Cref{alg:LP} and \Cref{algo:marginal-resolver} such that $\eta H(\Delta^R) < 1$.
%  \htodo{Is this necessary?}
  Then the output $\widehat{r}$ of \Cref{algo:marginal-resolver} satisfies
  \begin{align*}
    (1 + \widehat{\eps})^{-1} \leq \frac{\widehat{r}}{r} \leq 1 + \widehat{\eps},
  \end{align*}
  where $r=\frac{\mu_\Lambda(\sigma)}{\mu_{\Lambda}(\tau)}$ and $\widehat{\eps} := 5b^{-2}\eta H(\Delta^R) \cdot \eps$.
\end{lemma}

Note that our choices of $R$ in \eqref{eqn:select-R} and $\eta$ in \eqref{eq:eta} are stronger than requiring $5b^{-2}\eta H(\Delta^R)<1$.
This is because in the full algorithm, we need to rewrite each $R_w$ into a product of two marginal ratios so that the recursion can continue.
Thus, at each recursion step, the error first increases because of this product,
and then shrinks by \Cref{lem:error-bound}.

\Cref{lem:error-bound} is a direct consequence of the following two lemmas.
%The proof of \Cref{lem:marginal-error} relies on the following two lemmas about \Cref{alg:LP}.
%Let $r=\frac{\mu_S(\sigma)}{\mu_S(\tau)}$ be the marginal ratio, where the partial configurations $\sigma$ and $\tau$ on $S$ are what the LP is constructed for.
\begin{lemma}  \label{lem:completeness}
  Suppose \Cref{cond:LP} holds.
  If the input $r_-$ and $r_+$ to \Cref{alg:LP} satisfies $r_- \le r \le r_+$,
  then the LP is feasible.
\end{lemma}

\begin{lemma}  \label{lem:soundness}
  Suppose \Cref{cond:LP} holds with $\eps\le 3b^{-2}$.  
  If \Cref{alg:LP} returns true for parameters $r_-$ and $r_+$, then
  \begin{align*}
    (1+\hat{\eps})^{-1} r_- \le r \le (1+\hat{\eps})r_+.
  \end{align*}
\end{lemma}

The proofs of \Cref{lem:completeness} and \Cref{lem:soundness} are deferred to the next section, \Cref{sec:analysis}.

\begin{proof}[Proof of \Cref{lem:error-bound}]
  First suppose the binary search terminates early in \Cref{line:early-exit} of \Cref{algo:marginal-resolver}.
  In this case, both LP$(r_{low},m,\+T,\*R, \eps)$ and LP$(m,r_{upp},\+T,\*R, \eps)$ return true.
  \Cref{lem:soundness} implies that $r\le (1+\widehat{\eps})m$ and $r\ge (1+\widehat{\eps})^{-1}m$.
  Thus $m$ satisfies the desired approximation bound.
  
  Otherwise, only one of the two LPs is feasible until the while loop ends.
  By induction $r \in [r_{low},r_{upp}]$ throughout the while loop due to \Cref{lem:completeness}. 
  The binary search exits the while loop when $r_{upp} \le (1+\hat{\eps})^2 r_{low}$.
  Thus the output satisfies the desired approximation bound.
\end{proof}

\subsection{Error analysis of the LP}
\label{sec:analysis}

In this section, we prove \Cref{lem:completeness} and \Cref{lem:soundness}.
%We recall that both \Cref{lem:completeness} and \Cref{lem:soundness} assume that \Cref{cond:LP} holds.
%For convenience, we will use the same context as \Cref{cond:LP} and assume \Cref{cond:LP} holds through the rest of the section.
%We start with \Cref{lem:completeness}.

\begin{proof}[Proof of \Cref{lem:completeness}]
  For each node $w \in \+V(\+T)$ with $\mathrm{label}(w) = (\sigma^w,\tau^w,\Lambda^w,\ast)$, let $z_w$ be the probability that \Cref{alg:coupling} reaches the node $w$ in the coupling tree. 
  Furthermore, let $x_w = \frac{z_w \mu_\Lambda(\sigma)}{\mu_{\Lambda^w}(\sigma^w)}$ and $y_w = \frac{z_w \mu_\Lambda(\tau)}{\mu_{\Lambda^w}(\tau^w)}$.
  It suffices to verify that $(\*x,\*y)$ is a feasible solution to the LP.
  We verify the four sets of constraints one by one.

  \begin{itemize}
    \item Validity constraints: it is obvious that $x_w,y_w \ge 0$ and $x_{rt} = y_{rt} = z_{rt} = 1$.
    \item Recursive constraints:
      By symmetry, we only verify the first set.
      For each non-leaf node $w \in \+V(\+T)$ with $\mathrm{label}(w) = (\sigma^w,\tau^w,\Lambda^w,\emptyset)$, it holds that
      \begin{align*}
        \forall v \in S_R(u) \setminus \Lambda^w, c \in \mathrm{supp}(\mu^{\sigma^w}_v),\quad  \sum_{\substack{w' \in \+C(w):\\\mathrm{label}(w')=\tp{\sigma^w_{v \gets c},\ast, \ast, \ast} }} z_{w'}  
        = z_w\cdot\frac{1}{\ell(w)}\cdot\frac{\mu_{\Lambda^w \cup \{v\}}(\sigma^w_{v \gets c})}{\mu_{\Lambda^w}(\sigma^w)}.
      \end{align*}
      The left hand side is the total probability of reaching $w'\in\+C(w)$ whose first label is $\sigma^{w}_{v\gets c}$.
      This can only happen by first reaching $w$, with probability $z_w$,
      and randomly chosen $v$, with probability $\frac{1}{\ell(w)}$.
      Then, as \Cref{line:TV-coupling} of \Cref{alg:coupling} is a valid coupling,
      the probability of getting $\sigma^{w}_{v\gets c}$ is $\frac{\mu_{\Lambda^w \cup \{v\}}(\sigma^w_{v \gets c})}{\mu_{\Lambda^w}(\sigma^w)}$.
      This is exactly the right hand side.
      The recursive constraints then hold for $x_w=\frac{z_w \mu_\Lambda(\sigma)}{\mu_{\Lambda^w}(\sigma^w)}$.

    \item Leaf constraints:
      For each leaf node $w \in \+{L}(\+T)$ with $\mathrm{label}(w) = (\sigma^w,\tau^w,\Lambda^w,\ast)$, 
      we have 
      \begin{align*}
        \frac{x_w}{y_w} = r\cdot\frac{\mu_{\Lambda^w}(\tau^w)}{\mu_{\Lambda^w}(\sigma^w)}.
      \end{align*}
      By \Cref{cond:LP}, the term $\frac{\mu_{\Lambda^w}(\tau^w)}{\mu_{\Lambda^w}(\sigma^w)}$ is either exactly $R_w^{-1}$ (when $w$ is a good leaf) or approximated by $R_w^{-1}$ up to $1+\eps$ relative error (when $w$ is bad).
      As $r_-\le r\le r_+$, the leaf constraints hold. 
    \item Overflow constraints:
      By symmetry, we only verify the first set. For each non-leaf node $w \in \+V(\+T)$, %with $\mathrm{label}(w) = (\sigma^\star,\tau^\star,S^\star,\emptyset)$, it holds that
      \begin{align}\label{eqn:z-w}
        \sum_{\substack{w' \in \+C(w) \cap \+{BL}(\+T)}} z_{w'} \le \frac{b\eta}{\ell(w)} z_w.
      \end{align}
      This is because by \Cref{def:TI-decay}, the probability that \Cref{alg:coupling} reveals a disagreement is bounded by $\frac{\delta(R)}{\ell(w)}=\frac{b\eta}{\ell(w)}$. 
      Note that for any $w'\in\+C(w)$, $\Lambda^{w'}$ takes the form $\Lambda^w\cup\{v\}$ for some $v\in S_R(u)\setminus \Lambda^w$,
      and $\sigma^{w'}$ takes the form $\sigma^{w}_{v\gets c}$ for some $c\in\-{supp}(\sigma^{w}_v)$.
      By \Cref{def:marginal-bound}, the definition of the marginal lower bound $b$, we have
      \begin{align*}
        \frac{\mu_{\Lambda^{w'}}(\sigma^{w'})}{\mu_{\Lambda^w}(\sigma^{w})}=\frac{\mu_{\Lambda^w\cup\{v\}}(\sigma^{w}_{v\gets c})}{\mu_{\Lambda^w}(\sigma^{w})}\ge b.
      \end{align*}
      Therefore, by \eqref{eqn:z-w},
      \begin{align*}
        \sum_{w' \in \+C(w) \cap \+{BL}(\+T)} x_{w'} & = \sum_{w' \in \+C(w) \cap \+{BL}(\+T)} \frac{z_{w'}\mu_\Lambda(\sigma)}{\mu_{\Lambda^{w'}}(\sigma^{w'})}
        \le b^{-1}\sum_{w' \in \+C(w) \cap \+{BL}(\+T)} \frac{z_{w'}\mu_\Lambda(\sigma)}{\mu_{\Lambda^w}(\sigma^w)} \\
        & \le \frac{\eta}{\ell(w)} \cdot \frac{z_w\mu_\Lambda(\sigma)}{\mu_{\Lambda^w}(\sigma^w)} = \frac{\eta}{\ell(w)} x_w. \qedhere
      \end{align*}
  \end{itemize}
\end{proof}

From the proof above, one can observe that the overflow constraints are weaker than what the values we plug in for $x_w$ and $y_w$ really satisfy.
However, to write the exact constraints these values satisfy, some marginal probabilities are required.
There appears to be no efficient way to compute these marginal probabilities.
Thus, we choose to use the marginal lower bound to enforce a weaker set of constraints.

Next we show \Cref{lem:soundness}.
Suppose there exists a solution, say, $(\*x, \*y)$, to the linear program. 
For each $1 \le i \le \ell$ where $\ell = \abs{S_R(u) \setminus \Lambda}$, define
\begin{align*}
  \Gamma_{x,i} \defeq \sum_{\substack{w \in \+{BL}(\+T) :~\ell(w) = i-1 }} \mu_{\Lambda^w}(\sigma^w) x_w,
\end{align*}
where we use $\Lambda^w$ and $\sigma^w$ to denote the corresponding labels for $w$.
Intuitively, $\Gamma_{x,i}$ is the sum of $\mu_{\Lambda^w}(\sigma^w) x_w$ over all bad leaves of a certain depth. 
Furthermore, let
\begin{align*}
  \Gamma_{x,0} \defeq \sum_{\substack{w \in \+{GL}(\+T)}}  \mu_{\Lambda^w}(\sigma^w) x_w.
\end{align*}
Similarly, we define $\Gamma_{y,i}$ for $0 \le i \le \ell$ by replacing $x_w$ with $y_w$. 
Before proving \Cref{lem:soundness}, we derive some basic properties of $\Gamma_{\cdot, \cdot}$.

\begin{lemma}\label{lem:Gamma}
  Assuming \Cref{cond:LP},
  the following holds:
  \begin{enumerate}
    \item \label{item:Gamma-first}$\sum_{i=0}^\ell \Gamma_{x,i} = \mu_\Lambda(\sigma)$ and $\sum_{i=0}^\ell \Gamma_{y,i} = \mu_\Lambda(\tau)$;
    \item \label{item:Gamma-second}For all $1 \le i \le \ell$, $\Gamma_{x,i} \le \frac{\eta}{i} \sum_{j=0}^i \Gamma_{x,j}$ and $\Gamma_{y,i} \le \frac{\eta}{i} \sum_{j=0}^i \Gamma_{y,j}$;
    \item \label{item:Gamma-third}$r_- \Gamma_{y,0} \le \Gamma_{x,0} \le r_+ \Gamma_{y,0}$;
    \item \label{item:Gamma-fourth}For any $1\le i\le \ell$, $\tp{1+\eps}^{-2} r_- \Gamma_{y,i}  \le \Gamma_{x,i} \le \tp{1+\eps}^2 r_+ \Gamma_{y,i}$.
  \end{enumerate}
\end{lemma}

\begin{proof} 
  Note that for any non-leaf node $w \in \+V(\+T)$,% with $\mathrm{label}(w) = (\sigma^\star,\tau^\star,S^\star,\emptyset)$,
  \begin{align}
    \sum_{w'\in\+C(w)} \mu_{\Lambda^{w'}}\tp{\sigma^{w'}} x_{w'}
    &= \sum_{v \in S_R(u) \setminus \Lambda^w}\sum_{c \in \-{supp}(\mu^{\sigma^w}_v)}\sum_{\substack{w' \in \+C(w):\\ \mathrm{label}(w')=\tp{\sigma^w_{v \gets c},\ast,\ast,\ast}}} \mu_{\Lambda^w \cup \{v\} }\tp{\sigma^w_{v \gets c}} x_{w'} \notag\\
    &= \sum_{v \in S_R(u) \setminus \Lambda^w}\sum_{c \in \-{supp}(\mu^{\sigma^w}_v)} \frac{\mu_{\Lambda^w \cup \{v\}}\tp{\sigma^w_{v \gets c}} x_w}{\ell(w)} \notag\\
    &= \sum_{v \in S_R(u) \setminus \Lambda^w} \frac{\mu_{\Lambda^w}\tp{\sigma^w} x_w}{\ell(w)}=\mu_{\Lambda^w}(\sigma^w) x_w,\label{eq:prop-1}
  \end{align}
%  where we use $\Lambda^{w'}$ and $\sigma^{w'}$ to denote the corresponding labels for $w'$.
  where the second equality follows from recursive constraints of $\*x$, and the fourth equality follows from the definition of $\ell(w)$. 
  Recursively applying~\eqref{eq:prop-1} gives us that
  \begin{align*}
    \sum_{i=0}^\ell \Gamma_{x,i} = \mu_\Lambda(\sigma)x_{rt}=\mu_\Lambda(\sigma),
  \end{align*}
  as $x_{rt}$ is set to $1$ by the LP.
  A similar proof works for $\Gamma_{y,\cdot}$.
  The first item holds.

  For any non-leaf node $w \in \+V(\+T)$ %with $\mathrm{label}(w) = (\sigma^\star,\tau^\star,S^\star,\emptyset)$ 
  and $w'\in\+C(w)$,
  $\sigma^{w'}$ must have the form $\sigma^w_{v_{w'}\gets c_{w'}}$ for some $v_{w'}\in S_R(u)\setminus \Lambda^w$ and $c_{w'} \in \-{supp}(\mu^{\sigma^w}_{v_{w'}})$,
  and $\Lambda^{w'}$ must be $\Lambda^w\cup\{v_{w'}\}$.
  Thus,
  \begin{align*}
    \sum_{\substack{w' \in \+C(w) \cap \+{BL}(\+T) }} \mu_{\Lambda^w\cup\{v_{w'}\} }\tp{\sigma^w_{v_{w'}\gets c_{w'}}} x_{w'}
    &\le \sum_{\substack{w' \in \+C(w) \cap \+{BL}(\+T)}} \mu_{\Lambda^w}\tp{\sigma^w} x_{w'} 
    \le \frac{\eta}{\ell(w)} \mu_{\Lambda^w}(\sigma^{w}) x_w,
  \end{align*}
  where the first inequality follows from the fact 
  that $ \mu_{\Lambda^w \cup \{v\} }\tp{\sigma^w_{v \gets c}}\le \mu_{\Lambda^w}\tp{\sigma^w}$ for any $v\in S_R(u)\setminus \Lambda^w$ and $c \in \-{supp}(\mu^{\sigma^w}_{v})$,
  and the second inequality follows from the overflow constraints in \Cref{alg:LP}. 
  Also notice that $\ell(w)=\ell(w')+1$.
  It implies
  \begin{align*}
    \Gamma_{x,i} \le \frac{\eta}{i} \sum_{\substack{w \in \+{V}(\+T)\setminus \+L(\+T):~\ell(w) = i }} \mu_{\Lambda^w}(\sigma^w) x_w.
  \end{align*}
  On the other hand, recursively applying~\eqref{eq:prop-1}, we have %$\sum_{j=0}^{i-1} \Gamma_{x,j} = \sum_{\substack{\text{non-leaf } w:\\\ell(w) = i}} \mu_{\Lambda^w}(\sigma^w) x_w$, which implies
  \begin{align*}
    \sum_{j=0}^i \Gamma_{x,j}=\sum_{\substack{w \in \+{V}(\+T)\setminus \+L(\+T):~\ell(w) = i }} \mu_{\Lambda^w}(\sigma^w) x_w.
  \end{align*}
  A similar proof works for $\Gamma_{y,\cdot}$. The second item follows. %\todo{something seems off here. $\sum_{j=0}^i \Gamma_j$ should equal to the sum over non-leaf vertices.}
  
  For the last two items, by the leaf constraints and \Cref{cond:LP},
  \begin{align*}
    \Gamma_{x,0} &= \sum_{\substack{w \in \+{GL}(\+T)  }}  \mu_{\Lambda^w}(\sigma^w) x_w 
    \le \sum_{\substack{w \in \+{GL}(\+T)  }} \mu_{\Lambda^w}(\sigma^w) r_+ \cdot \frac{\mu_{\Lambda^w}(\tau^w)}{\mu_{\Lambda^w}(\sigma^w)} \cdot  y_w
    = r_+ \Gamma_{y,0};\\
    \Gamma_{x,i} &= \sum_{\substack{w \in \+{BL}(\+T) :~\ell(w) = i-1 }}  \mu_{\Lambda^w}(\sigma^w) x_w 
    \le \tp{1+\eps}^2 \sum_{\substack{w \in \+{BL}(\+T) :~\ell(w) = i-1 }}  \mu_{\Lambda^w}(\sigma^w) r_+ \cdot \frac{\mu_{\Lambda^w}(\tau^w)}{\mu_{S^w}(\sigma^w)} \cdot  y_w\\
    &= \tp{1+\eps}^2 r_+ \Gamma_{y,i}.
  \end{align*}
  Similarly, $\Gamma_{x,0} \ge r_- \Gamma_{y,0}$ and $\Gamma_{x,i} \ge \tp{1+\eps}^{-2} r_- \Gamma_{y,i}$.
\end{proof}

With the help of \Cref{lem:Gamma}, we are now ready to give the proof of \Cref{lem:soundness}.

\begin{proof}[Proof of \Cref{lem:soundness}]
  Let $\eps > 0$ be input to \Cref{alg:LP}.
  For $0\le\ell\le \Delta^R$, define
  \begin{align} \label{eq:error-sequence}
    \eps_\ell \defeq 5b^{-2}\eta H(\ell) \cdot \eps,
  \end{align}
  where the function $H(\ell) = \sum_{i=1}^\ell \frac{1}{i}$ is the harmonic sum.
  Clearly $\eps_\ell$ is increasing in $\ell$.

  To prove \Cref{lem:soundness}, by Item~\eqref{item:Gamma-first} of \Cref{lem:Gamma}, we only need to show that
  \begin{align}\label{eq:lem-9}
    \tp{1+\eps_{\ell}}^{-1} r_- \le \frac{\sum_{i=0}^{\ell} \Gamma_{x,i}}{\sum_{i=0}^{\ell} \Gamma_{y,i}} \le \tp{1+\eps_{\ell}} r_+.
  \end{align}
  We do an induction on $\ell$. By Item~\eqref{item:Gamma-third} in \Cref{lem:Gamma}, \eqref{eq:lem-9} holds for $\ell = 0$. 
  Now assume \eqref{eq:lem-9} holds for $\ell-1$.
  We prove the upper bound first. By induction hypothesis, 
  \begin{align*}
    \sum_{i=0}^{\ell-1} \Gamma_{x,i} \le \tp{1+\eps_{\ell-1}} r_+ \sum_{i=0}^{\ell-1} \Gamma_{y,i}
  \end{align*}
  Furthermore, by~Item~\eqref{item:Gamma-fourth} of \Cref{lem:Gamma}, $\Gamma_{x,\ell} \le \tp{1+ \epsilon}^2 r_+ \Gamma_{y,\ell}$.
  Therefore,
  \begin{align*}
    \sum_{i=0}^{\ell} \Gamma_{x,i} 
    &\le r_+ \tp{\tp{1+\eps_{\ell-1}} \sum_{i=0}^{\ell-1} \Gamma_{y,i} +  \tp{1+\eps}^2 \Gamma_{y,\ell}}
  \end{align*}
  We claim that
  \begin{align*}
    \tp{1+\eps_{\ell-1}} \sum_{i=0}^{\ell-1} \Gamma_{y,i} +  \tp{1+\eps}^2 \Gamma_{y,\ell} \le \tp{1+\eps_{\ell}}\sum_{i=0}^{\ell} \Gamma_{y,i},
  \end{align*}
  which would finish the proof.
  %By the assumption $\eta(\Delta^R) < 1$ in \Cref{lem:soundness} (inherited from \Cref{lem:error-bound}), it holds that $\ell_{\ell-1} < \eps_\ell = \eta H(\ell) \eps \leq \eta H(\Delta^R) \eps < \eps$.
  Our choice of parameters satisfy \eqref{eqn:select-R} and \eqref{eq:eta}, which implies that
  \begin{align*}
    \eps_{\ell-1}<\eps_\ell \le \eta H(\Delta^R) \eps < \eps.
  \end{align*}
  Thus, the claim is equivalent to 
  \begin{align*}
    \Gamma_{y,\ell} \le \frac{\eps_{\ell}-\eps_{\ell-1}}{2\eps+\eps^2 - \eps_{\ell}}\sum_{i=0}^{\ell} \Gamma_{y,i}.
  \end{align*}
  As $\eps\le 3b^{-2}$ and $0<b<1$, $2\eps+\eps^2-\eps_{\ell}\le5b^{-2}\eps$,
  we just need to show
  \begin{align*}
    \Gamma_{y,\ell} \le \frac{\eps_{\ell}-\eps_{\ell-1}}{5b^{-2}\eps}\sum_{i=0}^{\ell} \Gamma_{y,i} = \frac{\eta}{\ell}\sum_{i=0}^{\ell} \Gamma_{y,i}.
  \end{align*}
  This is just Item~\eqref{item:Gamma-second} of \Cref{lem:Gamma}.

  The lower bound in~\eqref{eq:lem-9} holds by a similar argument.
  Specifically, we can flip the role of $x$ and $y$ and replace $r_+$ with $\frac{1}{r_-}$ in the proof above to show the lower bound.
\end{proof}

\subsection{The full algorithm}
\label{sec:determistic-algo}
With the \Cref{lem:error-bound} in hand, we are now able to construct the full algorithm for \Cref{thm:main},
which is a recursive marginal estimator.
It uses \Cref{alg:coupling-tree} and \Cref{algo:marginal-resolver} recursively and do a truncation at some depth $k$.
%
%Note that these algorithms will use $R, \eta$ as parameters.
%To make sure the error decays exponentially along the recursions, we may set $\eta = b^{-1}\delta(R)$ ($b$ is the marginal lower bound and $\delta(R)$ is the total influence at $\partial B_R(u)$) and choose 
Recall that we have set the parameter $R$ to satisfy \eqref{eqn:select-R} and $\eta = b^{-1}\delta(R)$ in \eqref{eq:eta}.
Moreover, set the error parameter $\zeta_k$ of depth $k$ to be
\begin{align*}
  \zeta_k := \begin{cases}
    b^{-1}, & \text{ if } k = 0; \\
    2^{-k+1}, & \text{ if }k \geq 1.
  \end{cases}
\end{align*}
The full algorithm is given in \Cref{alg:marginal}.

\begin{algorithm}
    \caption{Recursive marginal estimator}
    \label{alg:marginal}
    \SetKwInOut{Input}{Input}
    \SetKwInOut{Output}{Output}
    \SetKwInOut{Parameters}{Parameters}
    \underline{Recursive-estimator} $(\sigma,\tau,\Lambda,k,u)$\;
    \Input{Partial configurations $\sigma,\tau \in [q]^\Lambda$ that only differ at a vertex $u \in \Lambda$, and $k \in \^{N}_{\ge 0}$}
    \Output{An estimate $\tilde{r}$}
    \Parameters{A positive integer $R$, a positive real $\eta$, and a function $\zeta_k$ for $k \in \^{N}_{\ge 0}$}
    $\+T \gets $ Coupling-Tree($\sigma,\tau$)\;
    \If{$k = 0$}{
        \Return 1\;
    }
    $\ell \gets \abs{S_R(u) \setminus \Lambda}$\;
    \If{$\ell = 0$}{
      \Return $r=\frac{\mu_{\Lambda}(\sigma)}{\mu_{\Lambda}(\tau)}$, computed by brute force;
    }
    \tcp{Calculate the marginal ratio of the leaves in $\+L(\+T)$ via recursion}
        \For{$w \in \+L(\+T)$}{
        $\tp{\sigma^w,\tau^w,\Lambda^w,D^w} \gets \mathrm{label}(w)$\;
                \eIf{$w \in \+{GL}(\+T)$}{
                  $R_{w} \gets$ the marginal ratio of $r = \frac{\mu_{\Lambda^w}(\sigma^w)}{\mu_{\Lambda^w}(\tau^w)}$ by brute force\;
                }
                {
                  $v\gets D^w$\;
                  $\gamma^w \gets \sigma^w_{u \gets \tau^w(u)}$\; 
                  %\tcp{Assertion: $\mu_{\Lambda^w}(\gamma^w) > 0$ by \eqref{eq:local2global} and $R > 0$, so that $X < \infty$}
                    $X \gets $ Recursive-estimator$\tp{\sigma^w,\gamma^w,\Lambda^w,k-1,u}$\label{line:X}\;
                    $Y \gets $ Recursive-estimator$\tp{\gamma^w,\tau^w,\Lambda^w,k-1,v}$\label{line:Y}\;
                    $R_w \gets X \cdot Y$\;
                }
        }
    \Return Marginal-estimator $(\+T, \*R, 2\zeta_{k-1} + \zeta_{k-1}^2)$\;
\end{algorithm}

Next we analyse the accuracy and the running time of \Cref{alg:marginal}.

\begin{lemma}[accuracy] \label{lem:accuracy}
  Let $\tilde{r}$ be the estimation given by \Cref{alg:marginal}, it holds that for $k \geq 0$,
  \begin{align*}
    (1 + \zeta_k)^{-1} \leq \tilde{r}/r \leq 1 + \zeta_k.
  \end{align*}
\end{lemma}
\begin{proof}
  We prove by induction on $k$.
  When $k = 0$, \Cref{alg:marginal} simply returns $1$.
  Let $\gamma = \sigma_{\Lambda\setminus \set{u}}$ be the common configuration between $\sigma$ and $\tau$.
  Then
  \begin{align*}
    \frac{\mu_{\Lambda}(\sigma)}{\mu_{\Lambda}(\tau)} & = \frac{\mu_{\Lambda}(\sigma)}{\mu_{\Lambda\setminus \set{u}}(\gamma)} \cdot \frac{\mu_{\Lambda\setminus\set{u}}(\gamma)}{\mu_{\Lambda}(\tau)} \leq b^{-1}.
  \end{align*}
  Similarly, $\frac{\mu_{\Lambda}(\sigma)}{\mu_{\Lambda}(\tau)} \ge b$.
  Since $\zeta_0=b^{-1}$, the base case holds.

  When $k \geq 1$, by the induction hypothesis, we know that $X$ and $Y$ obtained at \Cref{line:X} and \Cref{line:Y} in \Cref{alg:marginal} are estimates to $\frac{\mu_{\Lambda^w}(\sigma^w)}{\mu_{\Lambda^w}(\gamma^w)}$ and $\frac{\mu_{\Lambda^w}(\gamma^w)}{\mu_{\Lambda^w}(\tau^w)}$, respectively, both with $(1 + \zeta_{k-1})$ relative error.
  Note that here $\mu_{\Lambda^w}(\gamma^w)>0$.
  This is because $\gamma^w$ is locally feasible by its definition, 
  and as the spin system is permissive (recall \Cref{def:permissive}), by \eqref{eq:local2global}, $\gamma^w$ is globally feasible.
  Since $R_w = X \cdot Y$, this implies that the relative error of $R_w$ to $\frac{\mu_{\Lambda^w}(\sigma^w)}{\mu_{\Lambda^w}(\tau^w)}$ is
  \begin{align*}
    (1 + \zeta_{k-1})^2 = 1 + 2\zeta_{k-1} + \zeta_{k-1}^2,
  \end{align*}
  which means that $\*R$ satisfies a $(2\zeta_{k-1}+\zeta_{k-1}^2)$-error bound (recall \Cref{cond:LP}).

  For $k=1$, $(2\zeta_{0}+\zeta_{0}^2)\le 3b^{-2}$ and thus \Cref{lem:error-bound} applies.
  The lemma holds because in this case, by our choice of parameters in \eqref{eqn:select-R} and~\eqref{eq:eta},
  \begin{align*}
    5b^{-2}\eta H(\Delta^R) (2\zeta_{0} + \zeta_{0}^2) \le 15b^{-4}\eta H(\Delta^R) <1 = \zeta_{1}.
  \end{align*}

  For $k\ge 2$, by \eqref{eqn:select-R} and~\eqref{eq:eta}, again,
  \begin{align*}
    5b^{-2}\eta H(\Delta^R) (2\zeta_{k-1} + \zeta_{k-1}^2) \leq 15b^{-2}\eta H(\Delta^R) \zeta_{k-1} < \zeta_{k-1}/2 = \zeta_{k}.
  \end{align*}
  Thus the lemma follows from \Cref{lem:error-bound} as well.
\end{proof}

Now, we can finish the proof of \Cref{thm:main}. 
\begin{proof} [Proof of \Cref{thm:main}]
  First note that the assumptions of the marginal lower bound and coupling independence both hold with respect to an arbitrary conditioning.
  Thus, by standard self-reduction \cite{JVV86}, approximating $Z$ up to relative error $1+\eps$ can be reduced to approximate $n$ marginal probabilities up to relative error $1+\frac{\eps}{n}$.
  The latter task then can be reduced to approximate $q-1$ marginal ratios up to the same relative error $1+\frac{\eps}{n}$,
  for which we use \Cref{alg:marginal}.
  Suppose the running time of \Cref{alg:marginal} with relative error $1+\frac{\eps}{n}$ is $T_{\eps/n}$, 
  then the overall running time is $O(qnT_{\eps/n})$.

  By \Cref{lem:accuracy}, to achieve relative error $1+\frac{\eps}{n}$, we only need to pick $k = \ceil{\log_2(n/\eps)} + 1$.
  Denote
  \begin{itemize}
    \item by $T$ an upper bound of the cost of \Cref{alg:coupling-tree} and \Cref{algo:marginal-resolver} (i.e., the max cost of a single recursive step);
    \item and by $N$ maximum number of the nodes of the coupling tree.
  \end{itemize}
  Clearly the total cost of \Cref{alg:marginal} is bounded by $T\cdot \sum_{i=0}^k (2N)^i \leq 2T (2N)^{k}$.
  We claim that
  \begin{align*}
    T = O(k + \log b^{-1}) \cdot\-{poly}(N) \quad \text{and} \quad N \le 2(q^2\Delta^R)^{\Delta^R },
  \end{align*}
  which implies 
  \begin{align*}
    T_{\eps/n}\le 2T \cdot (2N)^{k} = \log b^{-1} (n/\eps)^{O(\Delta^R \log (q \Delta^R))}.
  \end{align*}
  By \Cref{thm:CI->sphere-CI}, we can take $\delta(R)=2C 2^{-\ceil{R/2C}}$,
  where $C$ is the coupling independence parameter.
  To satisfy \eqref{eqn:select-R}, $R=O(C(\log b^{-1}+\log C + \log\log\Delta))$ suffices.
%  Thus the exponent in the running time can be further simplified into $ \Delta^{O(C(\log b^{-1} + \log C  +\log\log\Delta))} \log q$.
  Thus the bound on the running time in terms of $b,C,q,\Delta$ is
  \begin{align*}
    T_{\eps/n} = \tp{\frac{n}{\eps}}^{\Delta^{O(C(\log b^{-1} + \log C  +\log\log\Delta))} \log q}.
  \end{align*}

  To show the claim,
  for $T$, it is direct to see that it would take $O(N)$ time to construct $\+T$ in \Cref{alg:coupling-tree}.
  Then, for \Cref{algo:marginal-resolver},
  the binary search step requires repeated calls to \Cref{alg:LP} to estimate the marginal ratio.
  The number of calls is at most $O(\log \frac{n}{b\eps}) = O(k + \log b^{-1})$. % = O(k)$ times, where the least equation holds for large $n$ because $b$ is a constant and $k = \Omega(\log n)$. 
  For each call to \Cref{alg:LP}, it costs $\-{poly}(N)$ time to construct and solve the LP.

  For $N$, each step of the coupling tree has at most $\Delta^R$ choices for the vertex $v$, and at most $q^2$ choices for the colours on $v$.
  This continues for at most $\Delta^R$ times, giving the bound on $N$.
\end{proof}

For \Cref{cor:main-colouring-high},
the colouring instance in \Cref{cor:main-colouring-high} has $O_{\Delta,q}(1)$-CI, proved in \cite[Theorem 5.10 and Lemma 5.13]{CLMM23}.\footnote{Theorem 5.10 in \cite{CLMM23} is stated for spectral independnce. However, \cite[Lemma 5.13]{CLMM23} proved the stronger coupling independence result.}
Thus \Cref{cor:main-colouring-high} follows from \Cref{thm:main} and \Cref{obs:permissive-lb}.
The generic marginal lower bound from \Cref{obs:permissive-lb} can be improved using \cite[Lemma 3]{GKM15},
which implies $b\ge q^{-1}\left( \frac{2}{3} \right)^{\Delta}$ in this setting.
Moreover, \cite[Theorem 5.9 and Lemma 5.14]{CLMM23} directly establish influence decay for colourings in high-girth graphs. One can directly use them to set the parameter $R$ in~\eqref{eqn:select-R} to obtain an algorithm with a better running time.

\section{Low disagreement coupling from Markov chain coupling}\label{sec:contractive-coupling}

In this section, we establish coupling independence from contractive coupling for Markov chains.
We note that previous results like~\cite{blanca2022mixing, liu2021coupling} have investigated how to establish spectral independence by using contractive couplings of Markov chains.
As a byproduct, their results already implicitly suggested how to establish coupling independence via contractive coupling.
However, the proofs of previous works~\cite{blanca2022mixing, liu2021coupling} rely heavily on Stein's method (see~\cite{bresler2019stein, reinert2019approximating}), which is analytical and arguably unintuitive.
In this section, we first give a direct and simple proof that contractive coupling can be used to establish coupling independence.

%Let $(\Omega, \-d)$ be a finite metric space.
%For any two distribution $\mu$ and $\nu$ on $\Omega$, the $1$-Wasserstein distance (W1-distance) with respect to the metric $\-d$ between $\mu$ and $\nu$ is defined as
%\begin{align} \label{eq:def-W1-dis}
%  \+W_{\-d}(\mu, \nu) := \inf_{\+C} \Ex_{(X,Y)\sim\+C}[\-d\tp{X,Y}],
%\end{align}
%where the infimum is taken over all the possible couplings $\+C$ between $\mu$ and $\nu$.
%For two $\Omega$-valued random variables $X, Y$ with distribution $\mu,\nu$, we may also use $\+W_{\-d}(X,Y)$ to denote $\+W_{\-d}(\mu,\nu)$.
%
Given a Markov chain $P$ and the current state $X$, we use $P(X)$ to denote the next (random) state of $P$.
Also, recall the Wasserstein distance from \Cref{def:Wasserstein}.

\begin{lemma} \label{lem:W1-via-coupling}
  Let $(\Omega, d)$ be a finite metric space.
  Let $C > 0$ and $\delta \in (0, 1)$ be two parameters.
  Let $P, Q$ be Markov chains on $\Omega$ with stationary distribution $\mu$ and $\nu$.
  Suppose the following conditions are satisfied:
  \begin{itemize}
  \item ($C$-disagreement) for $X \in \-{supp}(\mu)$, $\+W_{\-d}(P(X), Q(X)) \leq C$;
  \item ($Q$ has $\delta$-contraction) for $X \in \-{supp}(\mu), Y \in \-{supp}(\nu)$, $\+W_{\-d}(Q(X), Q(Y)) \leq (1-\delta)\, \-d\tp{X,Y}$.
  \end{itemize}
  Then, it holds that the W1-distance between $\mu, \nu$ can be bounded by $\+W_{\-d}(\mu, \nu) \leq C/\delta$.
\end{lemma}

\begin{proof}
  Let $X, Y$ be sampled from $\mu$ and $\nu$.
  Then we have $P(X) \sim \mu$ and $Q(Y) \sim \nu$, which implies 
  \begin{align*}
    \+W_{\-d}(X,Y) = \+W_{\-d}(P(X),Q(Y)) 
    &\leq \+W_{\-d}(P(X), Q(X)) + \+W_{\-d}(Q(X), Q(Y)) \\
    &\leq C + (1-\delta) \+W_{\-d}(X,Y).
  \end{align*}
  We finish the proof by rearranging terms.
\end{proof}

\begin{remark}
  Note that $P$ and $Q$ in \Cref{lem:W1-via-coupling} are Markov chains on $\Omega$.
  This does not require the support of $\mu$ and $\nu$ to be exactly $\Omega$.
  In fact, it is possible that $\-{supp}(\mu) \subset \Omega$ or $\-{supp}(\nu) \subset \Omega$,
  and states in $\Omega\setminus\-{supp}(\mu)$ or $\Omega\setminus\-{supp}(\nu)$ are transient.
%
%  Moreover, it is apparent from the proof that if we relax the second condition to $\+W_{\-d}(Q(X), Q(Y)) \leq (1-\delta)\, \-d\tp{X,Y}+C'$ for some constant $C'$,
%  the final bound just becomes $\frac{C+C'}{\delta}$.
\end{remark}

To establish coupling independence, 
one could apply \Cref{lem:W1-via-coupling} with $\mu^{\sigma}$ and $\mu^{\tau}$ where $\sigma$ and $\tau$ are two partial configurations such that $\dist(\mu,\nu)=1$, as in the definition of \Cref{def:CI}.
The first condition typically holds for any Markov chain with local moves.
This is because all new disagreement in one step must be near where $\sigma$ and $\tau$ disagrees.
The second condition is implied by the existence of contractive couplings, and there are plenty available in the literature.
We will see \Cref{lem:W1-via-coupling} in action in the applications next.
In fact, we will design an intermediate distribution and use triangle inequality in the next section. 

\subsection{Application: coupling independence for proper colourings}
In this section we show coupling independence for colourings using \Cref{lem:W1-via-coupling}.

\begin{theorem} \label{thm:CI-11/6-coloring}
  Let $G = (V,E)$ be a graph with maximum degree $\Delta$.
  Let $\mu$ be the uniform distribution over $q$-colouring of $G$.
  If either of the following holds:
  \begin{itemize}
    \item $\Delta \geq 125, q \geq 1.809\Delta$; or
    \item $\Delta \geq 3, q \ge (11/6-\epsilon_0)\Delta$ for some fixed parameter $\epsilon_0 \approx 10^{-5}$,
  \end{itemize}
  then $\mu$ satisfies $O_{\Delta,q}(1)$-coupling independence.
  Moreover, the CI constant is $\-{poly}(q\Delta)$.
\end{theorem}

To apply \Cref{lem:W1-via-coupling}, we consider flip dynamics, for which contractive couplings are known under the conditions of \Cref{thm:CI-11/6-coloring} \cite{chen2019improved,CV24}.
In fact, we need to consider the conditional distribution $\mu^\sigma$ for some (not necessarily proper) partial colouring $\sigma$ on $\Lambda\subset V$.
As mentioned in \Cref{sec:prelim}, the distribution~$\mu^\sigma$ is the same as the uniform distribution over a list colouring instance,
obtained by removing $\Lambda$ and letting $L_v=[q]\setminus \{\sigma(u)\mid u\in \Lambda\cap N_G(v)\}$.
Note that if $q\ge (1+\alpha)\Delta$ for some $\alpha>0$,
then for any list colouring instance obtained this way, $\abs{L_v}-\deg_G(v)\ge\alpha\Delta$.
Also, as long as $q\ge \Delta+1$, even if $\sigma$ is not proper, the induced list colouring instance still has solutions.
Namely the system is permissive.

We define the flip dynamics for list colouring instances obtained this way.
Let the remaining graph be $G=(V,E)$ and $\{L_v\}_{v\in V}$ be the lists.
Notice $L_v\subseteq [q]$ for all $v\in V$.
Let $X$ be a (not necessarily proper) colouring.
We say a path $w_1-w_2-\cdots-w_\ell$ is an $(X(w_1), c)$-alternating path from $w_1$ to $w_\ell$ if for all $i$, we have $X(w_i) \subseteq \set{X(w_1), c}$ and $X(w_i)\neq X(w_{i+1})$.
The \emph{Kempe component} (or Kempe chain) for $X$, a vertex $v \in V$, and a colour $c\in[q]$ is defined by: 
\begin{align*}
  S_X(v, c) \defeq \set{u \in V \; \left\vert \text{ there is a $(X(v),c)$-alternating path from $v$ to $u$} \right.}.
\end{align*}
Moreover, let $S_X(v,X(v))\defeq\emptyset$.
Let $\{p_\ell\}_{\ell\ge 1}$ be some parameters.
Given a colouring $X$, the flip dynamics updates $X$ to $X'$ as follows:
\begin{enumerate}
  \item pick a vertex $v \in V$ and a colour $c\in [q]$ uniformly at random;
  \item let $S=S_{X}(v,c)$ and $\ell =\abs{S}$;
  \item if the colouring on $S$ can be flipped, do so with probability $p_\ell/\ell$ to obtain $X'$;
  \item otherwise, let $X'=X$.
\end{enumerate}
By flipping, we mean changing the colours of all $c$-coloured vertices in $S$ to $X(v)$ and change the colours of all $X(v)$-coloured vertices in $S$ to $c$.
Note that the new colour may not be available in the lists of corresponding vertices.
In that case, we simply let $X'=X$.
Note that a Kempe component of size $\ell$ is flipped with probability at most $\frac{p_\ell}{q n}$.
Given a list colouring instance $(G,L)$, the above transition rule is defined for all $X \in [q]^V$ even if $X$ is improper for the list colouring instance~$(G,L)$.

Flip dynamics was first considered by Vigoda~\cite{Vig00},
who designed a set of parameters $\{p_\ell\}_{\ell\geq 1}$ and a contractive coupling with respect to the Hamming distance when $q > (11/6)\Delta$.
As a consequence, the flip dynamics is rapidly mixing.
Later, Chen, Delcourt, Moitra, Perarnau and Postle~\cite{chen2019improved} improved the rapid mixing regime to $q \geq (11/6 - \epsilon_0)\Delta$ for some $\epsilon_0 \approx 10^{-5}$.
They showed that a contractive coupling exists for the flip dynamics with some different $\{p_\ell\}_{\ell\geq 1}$ by considering either variable length coupling with Hamming distance or using an alternative metric.
Very recently, Carlson and Vigoda~\cite{CV24} showed that if $\Delta \geq 125$ and $q \geq 1.809\Delta$ (note that $1.809 < 11/6-\epsilon_0$),
there is another set of $\{p_\ell\}_{\ell\geq 1}$ ensuring contractive coupling, by considering a more refined metric.

These alternative metrics can be related to the Hamming distance.
Given the metric space $(\Omega, \-d)$, we say that the metric $\-d$ is $\alpha$-equivalent (to the Hamming distance) for some $\alpha \geq 1$, if for all $x, y \in \Omega$, it holds that
\begin{align} \label{eq:def-ell-equiv}
  \alpha^{-1} \dist\set{x, y} \leq \-d\set{x,y} \leq \alpha \dist\set{x,y}.
\end{align}
All metrics mentioned above are in fact $2$-equivalent.

Moreover, what we really consider is list colouring instances.
One can verify (see \cite[Theorem 4.12 and Section 4.6]{CV24}, \cite[Section 6 and Appendix E]{chen2019improved}, and \cite[Lemma 4.12]{blanca2022mixing}) that the aforementioned contractive couplings extend to list colourings as well,
summarised as follows.

\begin{proposition} [\cite{CV24}] \label{prop:coupling-charlie}
  Let $\Delta\ge 125$ and $q$ be two integers.
  There exists a sequence of parameters $\{p_{\ell}\}_{\ell \geq 0}$  satisfying $p_i = 0$ for all $i > 6$ such the the following holds.
  Given a list colouring instance $(G,L)$ such that for any $v\in V$, $L_v\subseteq [q]$ and $\abs{L_v}-\-{deg}_G(v)\in[0.809 \Delta,5/6\Delta]$, let $\Omega \subseteq [q]^V$ denote the set of proper list colourings for $(G,L)$,
  there is a $2$-equivalent metric $\-d$ on $\Omega$ such that the flip dynamics $P$ with  $\{p_{\ell}\}_{\ell \geq 1}$ satisfies the following:
  for any two proper list colourings $X,Y \in \Omega$,
  \begin{align}\label{eq:contract-1}
    \+W_{\-d}(P(X),P(Y)) \leq \left(1 - \frac{10^{-5}}{nq}\right) \-d (X,Y),
  \end{align} 
  where $n$ is the number of vertices in $G$.
\end{proposition}

%\begin{proposition} \label{prop:coupling-charlie}
%  Let $\Delta\ge 125$ and $q$ be two integers.
%  Given a list-colouring instance $(G,L)$ such that for any $v\in V$, $L_v\subseteq [q]$ and $\abs{L_v}-\-{deg}_G(v)\in[0.8089 \Delta,5/6\Delta]$,
%  there is a $2$-equivalent metric $\-d$ such that the flip dynamics with following parameters has $10^{-6}/n$-contraction:
%  \begin{align*}
%    p1 = 1,\; p_2 = 0.324,\; p_3 = 0.154,\; p_4 = 0.088,\; p_5 = 0.044,\; p_6 = 0.011 \text{ and } p_i = 0, \text{ for } i \geq 7.
%  \end{align*}
%\end{proposition}

\begin{proposition} [\cite{chen2019improved}] \label{prop:coupling-sitan}
  Let $\Delta\ge 3$ and $q$ be two integers, and $\epsilon_0 \approx 10^{-5} > 0$ be a constant.
  There exists a sequence of parameters $\{p_{\ell}\}_{\ell \geq 0}$  satisfying $p_i = 0$ for all $i > 6$ such the the following holds.  
  Given a list colouring instance $(G,L)$ such that for any $v\in V$, $L_v\subseteq [q]$ and $\abs{L_v}-\-{deg}_G(v)\ge (5/6-\eps_0) \Delta$, let $\Omega \subseteq [q]^V$ denote the set of proper list colourings for $(G,L)$,
  there is a $2$-equivalent metric $\-d$ on $\Omega$ such that the flip dynamics $P$ with  $\{p_{\ell}\}_{\ell \geq 1}$ satisfies the following:
  for any two proper colourings  $X,Y \in \Omega$,
  \begin{align}\label{eq:contract-2}
    \+W_{\-d}(P(X),P(Y)) \leq \left(1 - \frac{10^{-9}}{nq}\right)\-d(X,Y),
  \end{align} 
  where $n$ is the number of vertices in $G$.
\end{proposition}

%\begin{remark}[Domains of colourings]
%Given a list colouring instance $(G,L)$, there are two domains of colourings: (1) $\Omega = [q]^V$ and (2) $\-{supp}(\mu) = \{ \sigma \in \Omega \mid \sigma \text{ is a proper list colouring for } (G,L) \}$, where $\mu$ denotes the uniform distribution over all proper list colourings.
%It is clear that $\-{supp}(\mu) \subseteq \Omega$. 
%(2) \Omega_{L} = \{\sigma \in \Omega \mid \forall v\in V, \sigma(v) \in L(v) \}; and (3) 
%In  \Cref{prop:coupling-charlie} and \Cref{prop:coupling-sitan}, the metric $\-d$ is defined over the domain $\Omega = [q]^V$. In other words, for any colouring $X,Y \in [q]^V$, their distance $\-d(X,Y)$ is always well-defined. We only assume that the contraction results in~\eqref{eq:contract-1} and~\eqref{eq:contract-2} hold for all $X,Y \in \-{supp}(\mu)$, which is sufficient for our proof. 
%We remark that in~\cite{chen2019improved,CV24}, the contraction results holds even for $X,Y$ beyond $\-{supp}(\mu)$, because they both used path coupling~\cite{BubleyD97} which requires them to prove contraction even for improper list colourings. 
%\end{remark}

\begin{remark}[Domains of colourings]
Given a list colouring instance $(G,L)$, let $\Omega = \-{supp}(\mu) = \{ \sigma\mid \sigma\in[q]^V \text{ and }\sigma \text{ is a proper list colouring for } (G,L) \}$, where $\mu$ denotes the uniform distribution over all proper list colourings.
  %(2) \Omega_{L} = \{\sigma \in \Omega \mid \forall v\in V, \sigma(v) \in L(v) \}; and (3) 
  In  \Cref{prop:coupling-charlie} and \Cref{prop:coupling-sitan}, the metric $\-d$ is defined over  $\Omega$ and the contraction results in~\eqref{eq:contract-1} and~\eqref{eq:contract-2} hold for all $X,Y \in \Omega$, which is sufficient for our proof. 
  We remark that \cite{chen2019improved,CV24} actually proved stronger results, their metric is defined over a superset of $\Omega$ and the contraction results holds even for $X,Y$ beyond $\Omega$, because they both used path coupling~\cite{BubleyD97} which requires them to prove contraction even for improper list colourings. 
\end{remark}

Note that in \Cref{prop:coupling-charlie} and \Cref{prop:coupling-sitan},
we state the results in terms of W1-distance.
This is implied by the contractive couplings from \cite{CV24} and \cite{chen2019improved}.
More precisely, they show that for any $X,Y \in \Omega$, there exists a coupling of $P(X),P(Y)$ such that $\Ex[\-d(P(X),P(Y))] \leq (1-\frac{C}{nq}) \-d(X,Y)$, where $C=10^{-5}$ or $10^{-9}$ is the constant in~\eqref{eq:contract-1} and~\eqref{eq:contract-2}.
%the expected distance between the two copies after one step of the chain.
%\begin{proposition} \label{prop:coupling-sitan}
%  Let $\Delta\ge 3$ and $q$ be two integers, and $\epsilon_0 \approx 10^{-5} > 0$ be a constant.
%  Given a list-colouring instance $(G,L)$ such that for any $v\in V$, $L_v\subseteq [q]$ and $\abs{L_v}-\-{deg}_G(v)\ge (5/6-\eps_0) \Delta$,
%  there is a $2$-equivalent metric $\-d$ such that the flip dynamics with the following parameters has $10^{-9}/n$-contraction:
%  \begin{align*}
%    p_1 = 1,\, p_2 = \frac{185}{616},\, p_3 = \frac{1}{6},\, p_4 = \frac{47}{462},\, p_5 = \frac{9}{154},\, p_6 = \frac{2}{77} \text{ and } p_i = 0 \text{ for } i \geq 7.
%  \end{align*}
%\end{proposition}

Now we are ready to prove \Cref{thm:CI-11/6-coloring}.

\begin{proof}[Proof of \Cref{thm:CI-11/6-coloring}]
Let $\sigma$ and $\tau$ be two (not necessarily proper) partial colourings on $\Lambda \subseteq V$ such that they differ at only $v$. Let $\mu$ denote the unifrom distribution over all proper $q$-colurings in $G=(V,E)$. We goal is to bound $\+W(\mu^\sigma,\mu^\tau)$ to establish the CI for $\mu$.

Let $(G^{\sigma},L^\sigma)$ and $(G^{\tau},L^{\tau})$ be the two list colouring instances induced by $\sigma$ and $\tau$.
Note that $G^\sigma = G^\tau$ and $L^\sigma(u) \neq L^\tau(u)$ only if $u \in N(v)$, where $N(v)$ is the set of neighbours of $v$ in $G$.
We introduce a middle list colouring instance $(G',L')$ such that $G' = G^\sigma=G^\tau = (V',E')$, where $V' = V \setminus \Lambda$, and for all $u \in V'\setminus N(v)$, $L'(u) = L^\sigma(u) =L^\tau(u)$ and for all $u \in V' \cap N(v)$, $L'(u) = [q]$. 
Note that $\sigma$ and $\tau$ differ only at one vertex $v$. We use the following bound 
\begin{align}\label{eq3}
  \+W(\mu^\sigma,\mu^\tau) = 1 + \+W(\mu^\sigma_{V'},\mu^\tau_{V'}) \leq 1 + \+W(\mu^\sigma_{V'},\mu') + \+W(\mu',\mu^\tau_{V'}),
\end{align}
where $\mu'$ is the uniform distribution for list colouring $(G',L')$.

We show how to bound $\+W(\mu^\sigma_{V'},\mu')$. The bound for $\+W(\mu',\mu^\tau_{V'})$ can be obtained from the same proof. 
Let the parameters $\{p_\alpha\}_{\alpha\geq 1}$ and the metric $\-{d}$ be the same as in either \Cref{prop:coupling-charlie} or \Cref{prop:coupling-sitan} for the list colouring $(G',L')$. 
Specifically, if we assume the first condition in \Cref{thm:CI-11/6-coloring}, we use \Cref{prop:coupling-charlie}, otherwise, we use \Cref{prop:coupling-sitan}.
The space $\Omega$ in both propositions is the same one, $\Omega = \-{supp}(\mu') \subseteq [q]^{V'}$, which is the set of all proper list colourings for $(G',L')$.
Let $P$ and $Q$ be flip dynamics for $(G^\sigma,L^\sigma)$ and $(G',L')$, respectively, using the same $\{p_\alpha\}_{\alpha\geq 1}$.
Let $\mu_P = \mu^\sigma_{V'}$ and $\mu_Q = \mu'$ denote the uniform distributions of list colourings $(G^{\sigma},L^\sigma)$ and $(G',L')$ respectively, which are the stationary distributions of $P$ and $Q$ respectively.
By the definition of $(G',L')$, $\-{supp}(\mu_P) \subseteq \-{supp}(\mu_Q)$.
Both $P$ and $Q$ can be viewed as Markov chains over the space $\-{supp}(\mu_Q)$ because $P(X),Q(X) \in \-{supp}(\mu_Q)$ if $X \in \-{supp}(\mu_Q)$.
Let $(\-{supp}(\mu_Q), \-d)$ be the metric assumed in the two propositions. The contraction results in~\eqref{eq:contract-1} and ~\eqref{eq:contract-2} hold for $Q$.

Next, we use \Cref{lem:W1-via-coupling} with $P$ and $Q$ to bound $ \+W(\mu_P,\mu_Q) $.
For any colouring $X \in \-{supp}(\mu_P)$, notice that $p_\alpha > 0$ only for $\ell \le 6$,
  and the transitions of flip dynamics $P$ and $Q$ can be different only if a neighbour of $v$ is in the Kempe component.
  Let $N(v)$ denotes the set of neighbours of $v$ in $G$.
  This allows us to couple $P(X)$ and $Q(X)$ as follows:
  \begin{enumerate}
    \item in both copies, choose the same vertex $u$ and the same colour $c$ and let $S=S_{X}(u,c)$;
    \item if $\abs{S}\ge 7$, let $P(X)=Q(X)=X$;
    \item if $v$ is not adjacent to any vertex in $S$ in graph $G$, couple $P(X)$ and $Q(X)$ perfectly;
    \item otherwise, couple $P(X)$ and $Q(X)$ independently.
  \end{enumerate}
  If the coupling goes to the last step, the vertex $u$ must be within distance $6$ from $v$ in the graph $G$.
  There are at most $\Delta^7$ such choices.
  When that happens, there are at most $12$ new disagreement.
  Thus, $\+W(P(X), Q(X)) \leq \frac{12\Delta^7}{|V'|}$.
  As $\-{d}$ is $2$-equivalent, it implies
  \begin{align}\label{eq:bound1}
    \+W_{\-{d}}(P(X), Q(X)) \leq \frac{24\Delta^7}{|V'|}.
  \end{align}
  This verifies the first condition of \Cref{lem:W1-via-coupling}.

  For the second condition of \Cref{lem:W1-via-coupling},
  we need to bound the distance $\+W_{\-d}(Q(X),Q(Y))$ for $X \in \-{supp}(\mu_P)$ and $Y \in \-{supp}(\mu_Q)$. 
  By the definition of the instance $(G',L')$, it holds that $\-{supp}(\mu_P) \subseteq \-{supp}(\mu_Q)$, and thus $X \in \-{supp}(\mu_Q)$. 
  Moreover, $(G',L')$ has larger colour lists than $(G^\sigma,L^\sigma)$ and $(G^\tau,L^\tau)$.
  Thus $(G',L')$ satisfies the condition of either \Cref{prop:coupling-charlie} or \Cref{prop:coupling-sitan}, 
  which implies that $Q$ has $C/(q|V'|)$-contraction for some constant $C>0$ with respect to $\-{d}$.
  Thus, together with \eqref{eq:bound1}, we apply \Cref{lem:W1-via-coupling} with $\mu=\mu_P,\nu=\mu_Q$ and $\Omega = \-{supp}(\mu_Q)$ to derive
  \begin{align*}
    \+W_{\-{d}}(\mu^\sigma_{V'},\mu') = \+W_{\-{d}}(\mu_P,\mu_Q) \leq  \frac{24\Delta^7/|V'|}{C/(q|V'|)} = \mathrm{poly}(q\Delta).
  \end{align*}

  As $\-{d}$ is $2$-equivalent to the Hamming distance, we have $\+W(\mu^\sigma_{V'},\mu') = \mathrm{poly}(q\Delta)$. Since the same proof works for $\+W(\mu^\tau_{V'},\mu')$, the lemma follows from~\eqref{eq3}.
\end{proof}

In the proof above, all analysis except the first equation in ~\eqref{eq3} considers list colouring instances on the subgraph $G'= G[V']$.
The intermediate instance $(G',L')$ is introduced because we need to ensure that $X\in \-{supp}(\mu_Q)$ when applying \Cref{prop:coupling-charlie} or \Cref{prop:coupling-sitan}.
When verifying the first condition in \Cref{lem:W1-via-coupling}, we only use the fact that the flip chain is local and the metric $\-{d}$ is 2-equivalent. 
When verifying the second condition, we use the contraction results in \Cref{prop:coupling-charlie} or \Cref{prop:coupling-sitan} for the flip chain $Q$ over list colourings of $(G',L')$.

\Cref{cor:main-colouring} is then a direct consequence of \Cref{thm:main}, \Cref{thm:CI-11/6-coloring}, and the marginal lower bound \cite[Lemma 3]{LY13}.

\subsection{Application: coupling independence via the Dobrushin-Shlosman condition}

In this section we show \Cref{cor:main-DS}.
Recall the Dobrushin-Shlosman condition from \Cref{def:DS}.

\begin{theorem} \label{thm:DS-CI}
  Let $\delta \in (0,1)$ be a parameter.
  Suppose $\mu$ is the Gibbs distribution of a permissive spin system $(G,q,A_E,A_V)$ satisfying the Dobrushin-Shlosman condition with gap $\delta$.
  Then, $\mu$ satisfies $(\frac{\Delta}{\delta}+1)$-coupling independence.
\end{theorem}

The proof of \Cref{thm:DS-CI} follows from \Cref{lem:W1-via-coupling} and a well-known contractive coupling for the Glauber dynamics when the $1$-norm of the Dobrushin influence matrix $\rho$ is bounded (see~\cite{BubleyD97,hayes2006simple} and references therein).
Here Glauber dynamics is a well-studied Markov chain where in each step, we uniformly at random select a variable, and update it conditioning on the rest of the configuration.

\begin{proposition} \label{prop:doub-glauber-coupling}
  Let $\delta \in (0,1)$ be a parameter.
  Suppose $\mu$ is a distribution over $[q]^V$ satisfying the Dobrushin-Shlosman condition with gap $\delta$.
  Let $P$ be the Glauber dynamics for $\mu$.
  Then for every $X, Y \in [q]^V$, %it holds that
  \begin{align*}
    \+W_{\dist}(P(X), P(Y)) \leq \tp{1 - \frac{\delta}{n}} \dist\tp{X,Y}.
  \end{align*}
\end{proposition}
The above contractive coupling result holds for all $X,Y \in [q]^V$ because the definition of Dobrushin influence in~\eqref{eqn:defDob} considers all possible pinnings, including improper pinnings for $\mu$.

\begin{proof}[Proof of \Cref{thm:DS-CI}]
  Let $\sigma$ and $\tau$ be two (not necessarily feasible) partial configuration on $\Lambda \subseteq V$ such that they differ at only $v\in V$. Let $V' = V \setminus \Lambda$.
  As the system is permissive, $\mu^\sigma_{V'}$ and $\mu^\tau_{V'}$ are well-defined.
  To apply \Cref{lem:W1-via-coupling},
  let $P$ and $Q$ be Glauber dynamics for $\mu^\sigma_{V'}$ and $\mu^\tau_{V'}$, respectively.
  Consider the coupling where we always choose the same vertex $u \in V'$ to update, and optimally couple the updates at $u$.
  Clearly, only neighbours of $v$ could be the new disagreement.
  Thus,
  \begin{align*}
    \+W(P(X), Q(X)) \leq \frac{\Delta}{|V'|}.
  \end{align*}
  This verifies the first condition of \Cref{lem:W1-via-coupling}.

  For the second condition of \Cref{lem:W1-via-coupling},
  notice that the Dobrushin influence matrix $\rho^\tau$ for $\mu^\tau_{V'}$ is dominated by the corresponding matrix $\rho$ for $\mu$.
  Thus, $\Vert \rho^\tau\Vert_1 \le \Vert \rho\Vert_1 \le 1-\delta$.
  By \Cref{prop:doub-glauber-coupling},
  $Q$ has $\delta/{|V'|}$-contraction. % Here, we need Proposition 24 holds for all X,Y in [q]^V
  Note that $\sigma$ and $\tau$ differ only at a single vertex $v$.
  Thus, we can apply \Cref{lem:W1-via-coupling} to derive
  \begin{align*}
    \+W(\mu^{\sigma},\mu^{\tau}) &= \+W(\mu^{\sigma}_{V'},\mu^{\tau}_{V'}) + 1 \le \frac{\Delta}{\delta} + 1. \qedhere
  \end{align*}
\end{proof}

\Cref{cor:main-DS} directly follows from \Cref{thm:main}, \Cref{obs:permissive-lb}, and \Cref{thm:DS-CI}.

\ifdoubleblind
\else

\section*{Acknowledgement}
We thank Konrad Anand and Graham Freifeld for some helpful discussions at an early stage of this project.
We thank Charlie Carlson, Eric Vigoda, and Hongyang Liu for clarifying some questions regarding the flip dynamics and useful feedback.

\fi

\bibliographystyle{alpha}
\bibliography{refs.bib}

\appendix

\section{Heuristics behind the linear program to solve couplings}
\label{sec:heuristics}

In this section we provide some heuristics behind the linear programming approach of Moitra \cite{Moi19}, used in \Cref{sec:rec-marginal-resolver}.
Let $\sigma$ and $\tau$ be two partial configurations on $\Lambda$ such that they differ on one vertex, say, $v$.
Let $\Omega_\sigma$ be the set of states that are consistent with $\sigma$, and similarly for $\Omega_\tau$.
Suppose $\+C$ over $\Omega_\sigma\times\Omega_\tau$ is a coupling between $\mu^{\sigma}$ and $\mu^{\tau}$. 

As $\+C$ is a valid coupling, it must satisfy
\begin{align*}
  \forall \sigma'\in\Omega_{\sigma}, \quad \sum_{\tau'\in\Omega_{\tau}}\+C(\sigma',\tau') & =\mu^{\sigma}(\sigma'), \\
  \forall \tau'\in\Omega_{\tau}, \quad \sum_{\sigma'\in\Omega_{\sigma}}\+C(\sigma',\tau') &= \mu^{\tau}(\tau').
\end{align*}
Using $\mu^{\sigma}(\sigma')=\frac{\mu(\sigma')}{\mu_{\Lambda}(\sigma)}$ and $\mu^{\tau}(\tau')=\frac{\mu(\tau')}{\mu_{\Lambda}(\tau)}$,
we have
\begin{align*}
  \forall \sigma'\in\Omega_{\sigma}, \quad \sum_{\tau'\in\Omega_{\tau}}\+C(\sigma',\tau')\mu_{\Lambda}(\sigma) / \mu(\sigma') = 1,\\
  \forall \tau'\in\Omega_{\tau}, \quad \sum_{\sigma'\in\Omega_{\sigma}}\+C(\sigma',\tau')\mu_{\Lambda}(\tau) / \mu(\tau') = 1.
\end{align*}
This gives us a linear system, where we may treat $\+C(\sigma',\tau')\mu_{\Lambda}(\sigma) / \mu(\sigma')$ and $\+C(\sigma',\tau')\mu_{\Lambda}(\tau) / \mu(\tau')$ as variables.
This system is under constrained and the coupling is not unique,
while our goal is to solve $r=\frac{\mu_{\Lambda}(\sigma)}{\mu_{\Lambda}(\tau)}$.
To do so, notice that for any $(\sigma',\tau')\in\Omega_\sigma\times\Omega_\tau$,
\begin{align}  \label{eqn:heu-r}
  r=\frac{\+C(\sigma',\tau')\mu_{\Lambda}(\sigma) / \mu(\sigma')}{\+C(\sigma',\tau')\mu_{\Lambda}(\tau) / \mu(\tau')} \cdot \frac{\mu(\sigma')}{\mu(\tau')}.
\end{align}
Adding $r$ as a variable and \eqref{eqn:heu-r} to the system would make the system non-linear.
Moreover, it also makes the system less robust.
Instead, we introduce variables $x_{\sigma',\tau'}$ to represent $\+C(\sigma',\tau')\mu_{\Lambda}(\sigma) / \mu(\sigma')$,
$y_{\sigma',\tau'}$ to represent $\+C(\sigma',\tau')\mu_{\Lambda}(\tau) / \mu(\tau')$,
and $r^-$ and $r^+$ as guessed lower and upper bound for $r$.
Then consider the following set of linear equalities and inequalities:
\begin{align}\label{eqn:heu-LP}
  \begin{split}
  \forall \sigma'\in\Omega_{\sigma}, \quad & \sum_{\tau'\in\Omega_{\tau}}x_{\sigma',\tau'} = 1, \\
  \forall \tau'\in\Omega_{\tau}, \quad & \sum_{\sigma'\in\Omega_{\sigma}}y_{\sigma',\tau'} = 1, \\
  \forall (\sigma',\tau')\in\Omega_{\sigma}\times\Omega_{\tau}, \quad &r^- y_{\sigma',\tau'} \le x_{\sigma',\tau'}\cdot\frac{\mu(\sigma')}{\mu(\tau')} \le r^+ y_{\sigma',\tau'}. 
  \end{split}
\end{align}
Note that while we cannot compute $\mu(\sigma')$ or $\mu(\tau')$ easily, their ratio $\frac{\mu(\sigma')}{\mu(\tau')}$ is easy to compute exactly.

It is easy to see that if $r^-\le r\le r^+$, then the system \eqref{eqn:heu-LP} has a solution.
On the other hand, if the system \eqref{eqn:heu-LP} has a solution, we have that
\begin{align}\label{eqn:heu-ratio}
  r^-\le \frac{\sum_{(\sigma',\tau')\in\Omega_\sigma\times\Omega_\tau}x_{\sigma',\tau'} \mu(\sigma')}{\sum_{(\sigma',\tau')\in\Omega_\sigma\times\Omega_\tau}y_{\sigma',\tau'}\mu(\tau')} \le r^+.
\end{align}
Notice that
\begin{align*}
  \sum_{(\sigma',\tau')\in\Omega_\sigma\times\Omega_\tau}x_{\sigma',\tau'} \mu(\sigma') & = \sum_{\sigma'\in\Omega_\sigma} \mu(\sigma') \sum_{\tau'\in\Omega_\tau }x_{\sigma',\tau'}\\
  & = \sum_{\sigma'\in\Omega_\sigma} \mu(\sigma') = \mu_{\Lambda}(\sigma),
\end{align*}
where in the second line we used the first constraint in \eqref{eqn:heu-LP}.
Similarly for $\tau$.
Thus, \eqref{eqn:heu-ratio} implies that $r^-\le r\le r^+$.
In conclusion, the system \eqref{eqn:heu-LP} has a solution if and only if $r^-\le r\le r^+$.
Therefore, we can do a binary search to find a very accurate estimate to $r$ by repeatedly solving the LP \eqref{eqn:heu-LP}.

There is one issue with the above though, namely the system has an exponential size.
Moitra \cite{Moi19} considered constructing the coupling in a greedy way, instead of listing all final outcomes at once.
He greedily couples vertices one by one in an exploratory way, conditioning on previous choices at each step.
Each intermediate state gets its own $x$ and $y$ variables,
and the transition probabilities are reflected by linear constraints.

In fact, in Moitra's process, we can stop at any pair of intermediate partial configurations $\sigma_0$ and $\tau_0$ over ${\Lambda_0}$, and write down the corresponding LP.
The main issue of doing this is that there is no good way of computing $\frac{\mu_{\Lambda_0}(\sigma_0)}{\mu_{\Lambda_0}(\tau_0)}$.
Moitra's choice is to prioritise getting the same boundary between $\sigma_0$ and $\tau_0$.
If this is achieved, and the inside of the boundary has a logarithmic size, then $\frac{\mu_{\Lambda_0}(\sigma_0)}{\mu_{\Lambda_0}(\tau_0)}$ can be computed efficiently.
The key property for Moitra's process to succeed is to have an exponentially small (in the number of steps of the coupling process) probability of failing to get the same boundary between the two copies.
This property guarantees that one can truncate the coupling process at an logarithmic depth,
and maintain the size of the LP to be a polynomial in the input size.
To certify the exponential tail of failure probability, his LP involves linear constraints derived from local uniformity for each transition step,
which no longer holds in our setting.

The main innovation of our approach in \Cref{sec:rec-marginal-resolver} is that we do not try to efficiently compute $\frac{\mu_{\Lambda_0}(\sigma_0)}{\mu_{\Lambda_0}(\tau_0)}$ for intermediate states.
Instead, we use recursion.
To do so, notice that our partial coupling, \Cref{alg:coupling}, outputs $(\sigma_0,\tau_0)$ that either share the same boundary or differ by exactly $2$ vertices.
In the first case, $\frac{\mu_{\Lambda_0}(\sigma_0)}{\mu_{\Lambda_0}(\tau_0)}$ can be computed exactly and efficiently.
In the second case, there is a partial configuration $\rho$ such that both $(\sigma_0,\rho)$ and $(\rho,\tau_0)$ differ on only one vertex.
As $\frac{\mu_{\Lambda_0}(\sigma_0)}{\mu_{\Lambda_0}(\tau_0)}=\frac{\mu_{\Lambda_0}(\sigma_0)}{\mu_{\Lambda_0}(\rho)}\cdot\frac{\mu_{\Lambda_0}(\rho)}{\mu_{\Lambda_0}(\tau_0)}$,
we apply recursion to approximate both $\frac{\mu_{\Lambda_0}(\sigma_0)}{\mu_{\Lambda_0}(\rho)}$ and $\frac{\mu_{\Lambda_0}(\rho)}{\mu_{\Lambda_0}(\tau_0)}$.
Doing so apparently doubles the approximation error. 
However, this error occurs only in the second case, where the partial coupling exits early.
Luckily, the probability of early exits is $O(\delta(R) R\log\Delta)$.
As $\delta(R)$ decays exponentially with $R$, we can make this probability as small as we want.
There is one more wrinkle, that is, 
we cannot really write down linear constraints that exactly capture the early exit probability, because doing so would involve probabilities that we cannot compute efficiently.
Instead, we choose to use the marginal lower bound $b$ and the total influence bound in our linear program to give an upper bound of the early exit probability.
See the overflow constraints in \Cref{alg:LP}.
Eventually, we choose $R$ such that $\delta(R)$ absorbs $R\log \Delta$ together with some polynomial factors in $b^{-1}$.
This makes sure that the error decays by a constant factor at each recursive call.

\end{document}